\documentclass[acmsmall,screen,nonacm]{acmart}
\usepackage{etoolbox}
\newtoggle{ext} 
\settoggle{ext}{true} 
\usepackage{algorithm}
\usepackage[noend]{algpseudocode}
\usepackage{mathtools}
\usepackage{xcolor}
\usepackage{wrapfig}
\usepackage{multirow}
\usepackage{booktabs}
\usepackage{xurl}
\usepackage{array}
\usepackage{enumitem}
\usepackage{longtable}
\usepackage{thmtools}

\iftoggle{ext}{%
\newcommand{\refapp}[1]{\Cref{#1}}
\newcommand{\Crefapp}[1]{\Cref{#1}}
}{%
\usepackage{xr}
\externaldocument{app}
\newcommand{\refapp}[1]{\cite[Appendix \ref*{#1}]{ChengF2026}}
\newcommand{\Crefapp}[1]{\Cref*{#1}}
}

\usepackage{cleveref}

\def\tool{{\sc Mosaic}}

\newcommand{\ar}{\mathsf{ar}}

\DeclarePairedDelimiter{\tp}{\langle}{\rangle}
\DeclarePairedDelimiter{\sem}{\llbracket}{\rrbracket}
\newcommand{\Node}{\mathsf{Node}}
\newcommand{\Data}{\mathsf{Data}}
\newcommand{\VNode}{\mathcal{V}_\mathsf{Node}}
\newcommand{\VData}{\cV_\Data}
\newcommand{\isProc}{\mathsf{isProc}}
\newcommand{\nb}{N}
\newcommand{\prog}{\mathbf{P}}

\newcommand{\val}{\mu}
\newcommand{\sval}{\texttt{\textmu}}

\newcommand{\setcomp}[1]{\overline{{#1}}}
\newcommand{\QFAshcroft}{\textsc{QF-Ashcroft}}
\newcommand{\Ashcroft}{\textsc{Ashcroft}}
\newcommand{\init}{\mathsf{init}}
\newcommand{\error}{\mathsf{error}}
\newcommand{\htri}[3]{\left\{ #1 \right\}\; \allowbreak {#2}\; \allowbreak \left\{ #3 \right\}}
\newcommand{\qftype}{\mathsf{qftype}}

\newcommand{\Inv}{\mathit{Inv}}
\newcommand{\bool}{\mathsf{bool}}
\newcommand{\subst}{\textsc{Subst}}
\newcommand{\gform}{\mathsf{Global}}
\newcommand{\spec}{\mathbf{S}}
\newcommand{\Assertion}[1]{\textsc{Assertion}[{#1}]}
\newcommand{\Transition}[1]{\textsc{Transition}[{#1}]}
\newcommand{\id}[1]{%
  \ensuremath{\mathit{\mathcode`\-=`\-\relax#1}}}

\newcommand{\clsk}[2]{{\sf CL}(#1,#2)}
\newcommand{\substructures}{\mathsf{sub}}
\newcommand{\into}{\hookrightarrow}
\newcommand{\nv}{\nu}
\DeclareUrlCommand{\benchmark}{\footnotesize\urlstyle{tt}}

\theoremstyle{definition}
\declaretheorem[name=Example,numberwithin=section]{example}
\declaretheorem[name=Definition,numberwithin=section]{definition}

\declaretheorem[name=Remark]{remark}
\declaretheorem[name=Theorem,numberwithin=section]{theorem}
\declaretheorem[name=Lemma,sibling=theorem]{lemma}
\declaretheorem[name=Proposition,sibling=theorem]{proposition}
\declaretheorem[name=Corollary,sibling=theorem]{corollary}

\makeatletter
\def\@parfont{\bfseries}
\makeatother

\let\textAA\AA
\renewcommand{\AA}{\ifmmode\mathbb{A}\else\textAA\fi}
\newcommand{\BB}{\mathbb{B}}

\newcommand{\DD}{\mathbb{D}}

\newcommand{\LL}{\mathbb{L}}

\newcommand{\NN}{\mathbb{N}}

\let\textSS\SS
\renewcommand{\SS}{\ifmmode\mathbb{S}\else\textSS\fi}
\newcommand{\TT}{\mathbb{T}}
\newcommand{\UU}{\mathbb{U}}

\newcommand{\cC}{\mathcal{C}}

\newcommand{\cI}{\mathcal{I}}
\newcommand{\cJ}{\mathcal{J}}

\newcommand{\cP}{\mathcal{P}}

\newcommand{\cR}{\mathcal{R}}
\newcommand{\cS}{\mathcal{S}}
\newcommand{\cT}{\mathcal{T}}
\newcommand{\cU}{\mathcal{U}}
\newcommand{\cV}{\mathcal{V}}

\newcommand{\ff}{\mathsf{f}}

\makeatletter
\iftoggle{ext}{
  \AtBeginDocument{
    \settopmatter{printccs=true,printacmref=true}
  }
  \renewcommand{\@mkbibcitation}{%
  \bgroup\par\medskip\small\noindent
  \emph{This article is an extended version of the conference paper of the same title, which was accepted to CAV 2026.}
\par\egroup}
}{}
\makeatother

\begin{document}

\title{Complete Local Reasoning About Parameterized Programs Over Topologies}
\iftoggle{ext}{
  \subtitle{(Extended Version)}
}{}

\author{Ruotong Cheng}
\orcid{0009-0004-1857-7251}
\affiliation{%
  \institution{University of Toronto}
  \city{Toronto}
  \country{Canada}
}
\email{chengrt@cs.toronto.edu}

\author{Azadeh Farzan}
\orcid{0000-0001-9005-2653}
\affiliation{%
  \institution{University of Toronto}
  \city{Toronto}
  \country{Canada}
}
\email{azadeh@cs.toronto.edu}

\begin{abstract}
This paper investigates the algorithmic safety verification problem of infinite-state parameterized concurrent programs over a rich set of communication topologies. The goal is to automatically produce a proof of correctness in the form of a {\em universally quantified} inductive invariant, where the quantification is over the nodes in the topology. We illustrate that under reasonable assumptions on the underlying topology, the problem can be reduced to and solved as a {\em compositional} scheme, that is, the verification of the parameterized family is reduced to a set of {\em local} proofs, in a {\em complete} manner. We propose a verification algorithm and demonstrate through a set of benchmarks over several different topologies that our approach is effective in proving parameterized programs safe. 
\end{abstract}

\maketitle

\section{Introduction}\label{sec:intro}
Verifying the safety of parameterized programs is a fundamental and well-studied problem in computer-aided verification. This paper focuses on an instance of this problem that is specified by two distinguishing features: (1) The parameterized programs are {\em infinite-state} and defined over {\em a rich class of underlying communication topologies}. (2) Safety is formulated as the reachability of an error state that relates a constant number of threads.

In verification of concurrent programs, compositional verification, often called {\em thread-modular reasoning}, aims to construct a proof of correctness for the entire program by composing correctness proofs for individual threads. This is, however, known to be {\em incomplete} for reasoning about parameterized programs. In \cite{HoenickeMP2017}, the concept of {\em thread modularity at many levels} is proposed as an alternative. It is proved to be {\em complete} w.r.t. a class of proofs in the form of universally quantified invariants, called {\em Ashcroft Invariants}. Rather than focusing on one thread at a time, the {\em local} proof can be the proof of the parallel product of $k$ threads, for a fixed $k$, when an Ashcroft invariant with $k - 1$ quantifiers proves the program correct. This scheme, however, is deeply tied to the assumption that the program comprises replicated threads communicating through a {\em bounded number} of shared global variables. What is the nature of compositional verification of parameterized programs if they are over an arbitrary communication topology? 

We propose an answer to this question by using the inherent {\em symmetry} in such parameterized programs. Let us use an example to illustrate the core ideas. Consider the parameterized program given in Figure \ref{fig:pip}(a,b), which is a family of pipelines with $n$ processes (depicted as circle nodes) arranged on a line in linear order, where the first and last nodes are {\em distinguished} from the rest. Each pair of neighbouring processes shares two variables, an integer-typed {\tt data} and a Boolean-typed {\tt ready}, which reside on the square nodes in between. {\tt ready} being true indicates that {\tt data} has been computed and is ready to be read; all instances of {\tt ready} are initially {\tt false}. Additionally, the green (top) square node hosts two {\em global} counters {\tt t} and {\tt f}. Each process has an uninitialized  Boolean-typed variable {\tt pos}, which serves as a nondeterministic Boolean choice. Observe that there are {\em unboundedly many shared variables and local variables}, which is an important point of departure from classic models such as that of \cite{FarzanKP2015,HoenickeMP2017}. 

\begin{figure}[t]
\begin{center}
\includegraphics[scale=0.6,alt={(a) illustrates the topology for the pipeline, where the first and last nodes are distinguished from the rest, and every process is connected to a top node that hosts global variables. (b) is the code template for every process. (c) and (d) are two subprograms that are vital to the parameterized verification but not in the parameterized family.}]{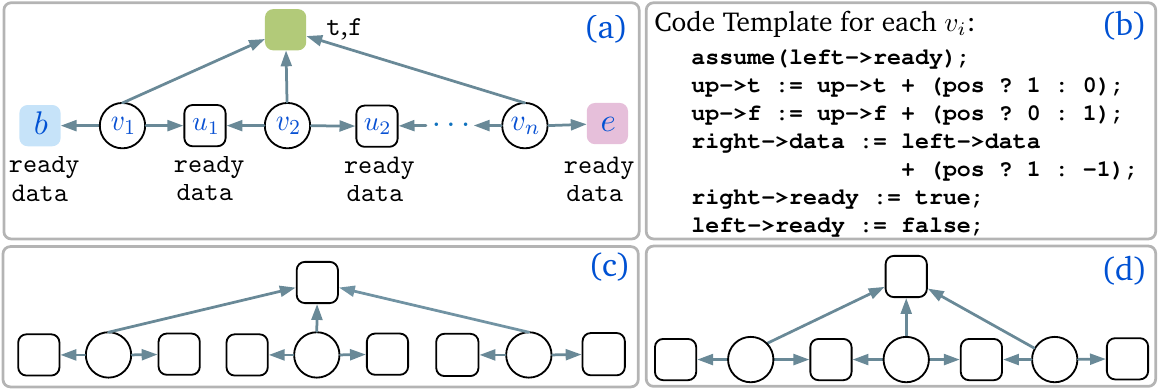}\vspace{-3pt}
\caption{Pipeline Example (a,b).  Small pipeline subprograms (c,d).\label{fig:pip}}
\vspace{-10pt}
\end{center}
\end{figure}


\def\ff{\mathop{\vartriangleright}}

Assume that each circle node atomically executes the given code template, where {\tt left} and {\tt right} refer to the corresponding square node neighbour of the process. 
The idea is that, depending on the nondeterministic value of {\tt pos}, the process increments or decrements the value before passing it along. Separately, the counts of true and false values for local {\tt pos} variables are recorded in global counters {\tt t} and {\tt f}. We want to prove the property below, which can be proved using a universally quantified inductive invariant: 
$$e\ff\texttt{ready} \implies e\ff\texttt{data} = b\ff\texttt{data} + \texttt{t} - \texttt{f}.$$


How can we discover this proof compositionally? The first key observation is that the smaller proofs do correspond to small subprograms (of programs in the parameterized family), but they do not necessarily have to be a single (or several) threads. For instance, Figure \ref{fig:pip}(c,d) illustrates two small subprograms with three processes each. In each case, the program includes processes {\em and} the {\em neighbourhoods} accessed by them. The distinction between the two programs is that in (d), the processes are strictly neighbours, but not in (c). The three-threaded programs in (c,d), however, {\em do not} belong in the pipeline parameterized family. It is clear why (c) is not a pipeline, and (d) does not have distinguished beginning and end nodes, and hence is not in the family. 

The thesis of this paper is that the proof for the pipeline parameterized family can be composed from the proofs of such subprograms. Intuitively, to construct the proof for any arbitrarily-sized pipeline, we break it into smaller pieces. If we have proofs for sufficiently many pieces that are distinct up-to-isomorphism, then we can use these proofs to give a proof for the entire pipeline program. This is fundamentally based on a {\em symmetry} argument formalized through an appropriate notion of local isomorphism. 

In Section \ref{sec:symmetry}, we prove a core result that if the parameterized family includes {\em sufficiently many such small programs}, then a generalized Ashcroft invariant proves the entire family correct {\em if and only if} it proves the small programs correct. Since programs in the family are instantiated from topologies in the family, this effectively puts a condition on the family of topologies. This has two orthogonal implications: (1) A candidate Ashcroft invariant can be verified using a bounded number of checks (Theorem \ref{thm:ashcroft-fin}), and (2) rather than searching for an Ashcroft invariant for the entire family, one can look for one for this bounded set. The former is a {\em cut-off} style theorem for verification. The latter is a key observation behind our compositional verification algorithm. This is inspired by the proof of relative completeness from \cite{pps}, where we linked a (different) proof system for parameterized programs to a generalization of Ashcroft invariants. We use these invariants, which accommodate arbitrary topologies, use a core observation from the proof, and extend it to prove our {\em cut-off} theorem.

The generalized Ashcroft invariants are not exactly {\em local} proofs, even for the small programs. The proof of the cut-off theorem relies on the insight that, without loss of generality, isomorphic neighbourhoods can share proofs, even if they are in different parts of the topology. In Section \ref{sec:chc-single}, we rely on this insight again to argue that each Ashcroft invariant can be syntactically rearranged (i.e., {\em normalized}) so that proofs of isomorphic neighbourhoods are {\em syntactically grouped}.  We refer to this transformation as {\em normalization}, and it streamlines the ingredients (function and predicate symbols) used in the matrix of the formula so that (1) there are distinct parts about the topology and the data manipulated by the program, and (2) using symmetry, one can enforce a bound on the size of terms in the matrix (Lemma \ref{lem:qftype-fin}). This results in a {\em bounded search space} for the ingredients of Ashcroft invariants, which can now be discovered through purely {\em local} proof searches for small programs. In other words, {\em normalization} produces a {\em complete template} for proof search. This search is then performed through the collection of {\em local} verification conditions as a system of constrained Horn clauses (CHC). As a consequence of (1), these local proofs become entirely oblivious to the program topology, which eliminates the need for axiomatizing the topology.
 
%
%
%
%

One can produce a proof for the entire parameterized family by using the procedure outlined above to collect constraints from all sufficiently many small programs of our cut-off theorem. There are, however, two obstacles to making this a pragmatic solution. First, we already mentioned that our pipeline example does not include the required small subprograms. In general, a typical user-friendly topology family (e.g. rings, lines, and grids) does not include these small subprograms. In \cite{pps}, we argue that one can {\em close} these topologies by adding the missing small programs while preserving the safety/unsafety status of the entire family. Second, even if one liberally closes the family, there is an enormous amount of redundancy in constraints collected from all necessary small programs.  In Section \ref{sec:encoding}, we give a construction that takes a standard user-friendly topology family, like the pipelines of Figure \ref{fig:pip}, and populates it with enough small programs to guarantee the required sufficiency conditions for compositional reasoning, and produces a CHC system with less redundancy from them. We take our treatment of redundancy further by introducing a set of {\em complete} optimizations in Section \ref{sec:opt}. These optimizations take advantage of the fact that our methodology of extracting a {\em complete template} for proof search from normalized Ashcroft invariants brings a degree of flexibility to this search. The normal form can change, and the set of templates for the search changes with it, but the guarantee of completeness remains. 

We have implemented our methodology in a tool called {\tool}, which takes as input a family of parameterized programs specified as a family of topologies and an assignment of nodes to a bounded number of code templates. 
Our experimental results show that the method is very effective in proving interesting properties of programs over a broad set of topologies. Furthermore, in \Cref{sec:tmk}, as a theoretical case study, we derive the algorithm of \cite{HoenickeMP2017} as a special case of our framework to demonstrate its robustness.   


\section{Background} \label{sec:background}

\paragraph{Structures and Maps Between Them}

We are concerned with structures in the context of first-order logic with equality. A \emph{vocabulary} $\cV$ is a set of predicate symbols and function symbols, each with a certain \emph{arity} given by a global map $\ar$. A \emph{$\cV$-structure} $\AA$ is defined as an \emph{underlying set} with an \emph{interpretation} of the symbols in $\cV$. We generally use $\AA$ to also refer to the underlying set of $\AA$.

Let $\AA, \BB$ be two $\cV$-structures. An \emph{embedding} $f : \AA \hookrightarrow \BB$ is an injection that preserves and reflects the interpretation of all predicate and function symbols. An \emph{isomorphism} is a surjective embedding. An \emph{automorphism} is an isomorphism from a structure to itself. If there is an embedding from $\AA$ into $\BB$, we say $\AA$ is an \emph{embedded substructure} of $\BB$; if furthermore, the underlying set of $\AA$ is a subset of that of $\BB$ and the inclusion map is an embedding, we say $\AA$ is a \emph{substructure} of $\BB$. The \emph{substructure generated} by a subset $A$ of $\AA$ is the smallest substructure containing $A$ that is closed under applying functions from $\cV$.


\paragraph{Constrained Horn Clauses}

For brevity, we assume the background theory (for program data) is given by a $\VData$-structure $\DD$. Let $\cR$ be a set of predicate symbols representing unknowns, disjoint from $\VData$. A \emph{constrained Horn clause} (CHC) is a formula in the vocabulary $\VData \cup \cR$ that is of the form $B_1 \land \cdots \land B_n \land C \Rightarrow H$, where every $B_i$ is an application $p(x_1, \ldots, x_{\ar(p)})$ of an unknown predicate symbol $p \in \cR$ to free variables $x_1, \ldots, x_{\ar(p)}$, $C$ is a $\VData$-formula, and the \emph{head} $H$ is either $\bot$ or an application as in $B_i$.

A \emph{CHC system} $\cC$ is a finite set of CHCs. It is \emph{satisfiable} if there exists an interpretation for each unknown predicate symbol that makes the universal closure of each clause in $\cC$ valid modulo $\DD$. A (syntactic) \emph{solution} to $\cC$ maps every unknown $p$ in $\cC$ to a $\VData$-formula with $\ar(p)$ free variables such that the subsets defined by these formulas constitute a satisfying interpretation.

\section{Parameterized Programs and Their Proofs}
In this section, we formalize concurrent programs over individual topologies and then parameterized programs over parameterized families of topologies. We formalize a class of safety proofs for such parameterized programs in the form of universally quantified inductive invariants.

\subsection{Concurrent Program Over Topology}

In our setup, a \emph{topology} is a finite logical structure $\TT$ over a vocabulary $\VNode$, whose underlying elements are called \emph{nodes}. For a set of nodes $V$, we define the \emph{neighbourhood} $\nb(V)$ as the substructure generated by the nodes in $V$. For example, Figure \ref{fig:pip}(d) illustrates the neighbourhood of two consecutive nodes in a pipeline, and Figure \ref{fig:pip}(c) that of two non-neighbouring nodes. We sometimes identify a tuple of nodes $\vec{v} \in \TT^k$ with the set $\{ v_1, \ldots, v_k \}$ and write $\nb(\vec{v})$.

We assume there is a variable on every node, which can be a structure consisting of multiple variables, standing in for variables local to a process or shared between neighbouring processes. Intuitively, a process that runs on a node only has access to the variables in its neighbourhood. 
For brevity, we fix a single data domain $\DD$ over a vocabulary $\VData$ and assume all variables have domain $\DD$. 

\paragraph{States and Transitions}
For any node $v$, we say a valuation $\nb(v) \to \DD$ is a \emph{local state} of $v$. A valuation $\TT \to \DD$ is a \emph{global state}. 
A \emph{concurrent program $\prog$ over a topology $\TT$} can be viewed as a transition system that captures the changes in the local states, which together comprise changes in the global state. 

Formally, $\prog$ is a pair $\tp{\TT, \sem{-}}$, where $\sem{-}$ maps every node $a$ to a binary relation on its local states, which induces a transition relation over the global states by declaring that variables outside the neighbourhood are unaffected: for any global states $\val, \val'$, we have $\val \xrightarrow{v} \val'$ if $(\val|_{N(v)}, \val'|_{N(v)}) \in \sem{v}$ and $\val|_{\setcomp{N(v)}} = \val'|_{\setcomp{N(v)}}$. We say $\val \to \val'$ if $\val \xrightarrow{v} \val'$ for some $v \in \VNode$.

We consider the vocabulary $\VNode$ (resp. $\VData$) to have a single sort $\Node$ (resp. $\Data$). Let $\sval$ and $\sval'$ be two formal symbols of type $\Node \to \Data$. We define a set $\Transition{\TT}$ consisting of formulas of the form $\phi[\sval(\vec{u}), \sval'(\vec{w})]$ (the notation $\sval(\vec{u})$ abbreviates the list $\sval(u_1), \ldots, \sval(u_{|\vec{u}|})$, and similarly for $\sval'(\vec{w})$), where $\phi$ is a $\VData$-formula and $\vec{u}, \vec{w}$ are tuples over $\TT$. We assume every $\sem{v}$ is definable by such a formula where the nodes are from $\nb(v)$.

\paragraph{Safety Specification}
We assume that the \emph{safety} of a concurrent program is specified through a pair of functions $\tp{\init, \error}$, where $\init$ and $\error$ map each node $v \in \TT$ to a set of local states of $v$.\footnote{To simplify presentation, we assume $\error$ is a unary function in our formalism. It is straightforward to account for an $m$-ary $\error$ function, where $m$ is no more than the width of the Ashcroft invariant to search for (see \Cref{sec:gai}).} For a global state $\val$, we say $\val$ satisfies $\init$ if $\val|_{\nb(v)} \in \init(v)$ for all $v \in \TT$, and $\val$ satisfies $\error$ if $\val|_{\nb(v)} \in \error(v)$ for some $v \in \TT$. A program $\prog$ is safe w.r.t. $\tp{\init, \error}$ if for any global states $\val, \val'$ such that $\val \to \val'$, if $\val$ satisfies $\init$, then $\val'$ does not satisfy $\error$.

The class of formulas $\Assertion{\TT}$ is defined similarly to $\Transition{\TT}$, except that the symbol $\sval'$ does not appear. We assume that every $\init(v)$ and $\error(v)$ is definable by such a formula where the nodes are from $\nb(v)$.

We abuse notations slightly and also use $\init(v)$, $\error(v)$, and $\sem{v}$ to refer to the formula that defines each.

\subsection{Parameterized Program Over a Family of Topologies} \label{sec:param}

Parameterized programs are generally defined over a notion of {\em symmetry}. That is, two programs are considered to be part of the same family when they are in some sense symmetric. In other words, we want symmetric nodes to have symmetric semantics. 
To formalize this, we use a {\em local} notion of symmetry defined based on the concept of a {\em local isomorphism}, inspired by \cite{pps}. 

\paragraph{Symmetry Equivalence} For structures $\TT$ and $\TT'$ over the same vocabulary, given tuples $\vec{u} \in \TT^d$ and $\vec{v} \in \TT^{d'}$, we say any isomorphism $\beta : \nb(\vec{u}) \to \nb(\vec{v})$ is a \emph{local isomorphism} (from $\vec{u}$ to $\vec{v}$) if it maps $\vec{u}$ to $\vec{v}$ (component-wise). \emph{Local isomorphisms} induce an equivalence relation $\simeq$, called \emph{local symmetry}, on tuples over $\TT$ and $\TT'$: we have $\vec{u} \simeq \vec{v}$ if there is a local isomorphism from $\vec{u}$ to $\vec{v}$. For example, in the pipeline topology of Figure \ref{fig:pip}(a), any two circle nodes are equivalent up to symmetry, as long as they are not the distinguished begin/end circle node; those two nodes are in two other symmetry equivalence classes.

The local isomorphism $\beta : \nb(\vec{u}) \to \nb(\vec{v})$ induces a map $\beta^* : (\nb(\vec{v}) \to \DD) \to (\nb(\vec{u}) \to \DD)$ between local states by $\mu \mapsto \mu \circ \beta$. We lift $\beta^*$ to (sets of) tuples of valuations in the natural way. In other words, symmetric nodes have symmetric semantics, which means that proof arguments from one generalize to the other up to symmetry. 

\paragraph{Parameterized Programs} 
Given a vocabulary $\VNode$ and a data domain $\DD$, a \emph{parameterized (concurrent) program} is an infinite family $\cP$ of programs over $\VNode$-topologies such that for every pair of programs $\tp{\TT, \sem{-}}$ and $\tp{\TT', \sem{-}'}$ in the family, for any $v \in \TT$ and $v' \in \TT'$ such that $v \simeq v'$ via $\beta$ (namely $v' = \beta(v)$), we have $\sem{v} = \beta^*\sem{v'}$. We write $\cP$ as an indexed set $\{ \prog_i \}_{i\in\cI}$, where $\prog_i = \tp{\TT_i, \sem{-}_i}$.

\paragraph{Safety Specification} 
The safety of a parameterized program is expressed by a \emph{parameterized specification} $\{\tp{\init_i, \error_i}\}_{i \in \cI}$, where the pair $\tp{\init_i, \error_i}$ is a safety specification for $\prog_i$. Naturally, we expect the specification to also be symmetric: Formally, for any $i, j \in \cI$, $u \in \TT_i$, and $v \in \TT_j$ such that $u \simeq v$ via $\beta$, we have $\init_i(u) = \beta^* \init_j(v)$ and $\error_i(u) = \beta^* \error_j(v)$. We say $\cP$ is safe if $\prog_i$ is safe w.r.t. $\tp{\init_i, \error_i}$ for every $i \in \cI$.

\subsection{Generalized Ashcroft Invariants} \label{sec:gai}

We adapt a class of universally quantified invariants for a generic set of topologies that appear in \cite{pps} with the simplifying assumption that program counter variables are available, and the control flow can be encoded through them; this is without loss of generality. The complete formal definition is recalled in \refapp{app:ash}. In short, we define \QFAshcroft{} as a fragment of first-order formulas over the vocabulary $\VNode \uplus \VData \uplus \{ \sval \}$ of two sorts, $\Node$ and $\Data$, where the special symbol $\sval$ is typed $\Node \to \Data$, such that no quantification is over $\Node$ and every free variable is of sort $\Node$. A generalized Ashcroft assertion of \emph{width} $k$ is a sentence of the form $\forall \nv_1, \ldots, \nv_k.\, \varphi[\nv_1, \ldots, \nv_k]$, where $\varphi \in \QFAshcroft$ and $\nv_1, \ldots, \nv_k$ are $\Node$-sorted variables.

The truth value of a (generalized) Ashcroft assertion is given by a $\VNode$-topology $\TT$ and a global state $\val : \TT \to \DD$, interpreting the symbol $\sval$ (recall that we fix the data domain). For a program $\prog$ and a safety specification $\tp{\init, \error}$, a \emph{generalized Ashcroft invariant} is an Ashcroft assertion that is {\em inductive} w.r.t. $\prog$ in the standard sense and \emph{sufficient} to establish safety. $\Phi$ is an Ashcroft invariant for a parameterized program w.r.t. a parameterized specification if it is an invariant for every member of the family.

\section{A Cut-Off Theorem for Ashcroft Invariants}\label{sec:symmetry}
The key observation of this section is as follows. Under certain conditions about the underlying family of topologies, an Ashcroft assertion for {\em a bounded number of programs} in a parameterized family is also an Ashcroft invariant for the entire parameterized family. This is formally stated as \Cref{thm:ashcroft-basis} at the end of this section. This can be viewed as a {\em cut-off} style theorem for Ashcroft invariants over these topologies. It is intuitively clear that the conditions required for such a theorem to hold have to do with the amount of {\em symmetry} present in the topology family. In some sense, we want the finite set of programs to enumerate all possible different {\em scenarios} that one may observe in any instance of the family of any size. In \cite{pps}, such a condition is defined, albeit used for a different purpose:

\begin{definition}[\cite{pps}] \label{def:topo-basis}
    Let $k \in \NN_+$ and $\cT$ be a family of topologies over a vocabulary $\cV$. A \emph{basis of rank $k$} for $\cT$ is a set $\cS$ of $\cV$-topologies where for any $\TT \in \cT$ and $\vec{v} \in \TT^k$, there is some $\SS \in \cS$ and embedding $f : \SS \hookrightarrow \TT$ such that $\vec{v} \in f(\SS)^k$.
\end{definition}

Recall the Pipeline example of \Cref{fig:pip}. Let us assume $k = 3$. If we pick three arbitrary {\em middle} circle nodes from the pipeline, there are three possibilities, depending on whether any pair is consecutive or not. \Crefapp{fig:pip}(c,d) illustrates two of them. A full basis of rank 3 with 15 members is listed in \refapp{app:pipeline}. Since the {\em first} and {\em last} nodes of the pipeline are distinguished, there are other choices of three {\em arbitrary} nodes beyond choosing only {\em middle} points.  


We can now state the following new theorem that relates Ashcroft invariants of a parameterized program to Ashcroft invariants of a \emph{finite} parameterized program, specifically the one over the basis:

\begin{theorem} \label{thm:ashcroft-basis}
    Let $\cP = \{ \prog_i \}_{i\in\cI}$ be a parameterized program over a topology family $\cT = \{ \TT_i \}_{i\in\cI}$. For any $k \in \NN$, if $\cT$ contains a basis $\{ \TT_j \}_{j\in\cJ}$ of rank $k + 1$ as a subset, then for any $\Phi$ of width $k$ and parameterized safety specification, $\Phi$ is an Ashcroft invariant for $\cP$ if and only if it is an Ashcroft invariant for $\{ \prog_j \}_{j\in\cJ}$.
\end{theorem}

We give a formal proof in \refapp{app:proofs}. To get a sense of why the theorem holds, and its connection to the {\em rank of the basis} for the topology family, we give an informal argument for the nontrivial direction here. Assume we have an Ashcroft invariant $\Phi$ for the basis $\{ \TT_j \}_{j\in\cJ}$, but $\Phi$ is not an Ashcroft invariant of the parameterized family. Then, it must not be an invariant of a specific member based on some topology $\TT_i$. Assuming the problem is not with initialization or safety of the invariant, there must be a specific statement of a specific node $v$ violating the inductivity of $\Phi$, which must be witnessed by a concrete tuple of nodes $\vec{u}$ of length $k$. By the premise and Definition \ref{def:topo-basis}, we know that some substructure of $\TT_i$ containing $v$ and $\vec{u}$ can be identified with a topology in the basis, which is a contradiction, because this violation must be observed in that member of the basis against our assumptions.

Observe that Theorem \ref{thm:ashcroft-basis} also indicates that checking whether a given Ashcroft assertion is indeed an invariant can be done through a bounded number of checks under the same conditions. In Section \ref{sec:ht}, we further argue (see Proposition \ref{prop:valid-entail}) that these checks are limited to entailments in the data domain $\DD$ and hence {\em independent} of the topology family. In a sense, the topology family through the high level view of Theorem \ref{thm:ashcroft-basis} determines a set of such checks, but any individual check is oblivious to it. 

\begin{theorem} \label{thm:ashcroft-fin}
    Let $k \in \NN$ and $\cT$ be a topology family such that a finite basis of rank $k+1$ contained in $\cT$ is given. For any Ashcroft assertion $\Phi$ of width $k$, parameterized program over $\cT$, and parameterized safety specification, one can compute a $\VData$-sentence that is valid modulo $\DD$ iff $\Phi$ is an Ashcroft invariant.
\end{theorem}
\section{From Ashcroft Invariants to Local Proofs}\label{sec:chc-single}

Theorem \ref{thm:ashcroft-basis} links the proofs of small programs over the {\em basis} of the topology to the proof of the entire parameterized family. These proofs, however, are in the form of Ashcroft invariants. In this section, we shift gears to local proofs, which are algorithmically searchable. 
We first introduce a normalization procedure for Ashcroft invariants, and then model the resulting {\em local} verification conditions as a set of Hoare triples. Finally, we state the proof search problem for a single program in the topology as a CHC system. 

\subsection{Normalizing Ashcroft Assertions}\label{sec:nai}

A free-form Ashcroft invariant of width $k$ can have an arbitrarily sized and shaped matrix, which can mix and match terms on the topology family and the underlying data. We use the underlying symmetry of the topology family to argue that any Ashcroft invariant can be rearranged in a way that would syntactically separate the terms about the topology family and the data, and further put a bound on the size of the universe of terms that can appear in it.



We start with two related key observations: (1) Every {\em quantifier-free formula} represents the union of a number of {\em local symmetry equivalence classes}, and (2) Under appropriate conditions, every {\em local symmetry equivalence class} can be represented by a single {\em quantifier-free formula}. 

There is a well-established concept of \emph{types of quantifier rank $r$} in model theory (e.g., see \cite[Sec. 3.4]{Libkin2004}), where a special case would be when the rank $r = 0$, and hence the focus is on quantifier-free formulas. Consider a tuple $\vec{v} \in  \TT^k$ where $\TT$ is a $\cV$-structure. The \emph{quantifier-free $k$-type of $\vec{v}$ over $\TT$}, denoted $\qftype(\TT, \vec{v})$, is the set of quantifier-free $\cV$-formulas $\phi$ such that $\TT \models \phi(\vec{v})$. We say $(\TT, \vec{v})$ \emph{realizes} $\qftype(\TT, \vec{v})$. Quantifier-free $k$-types precisely capture the quotient $\TT^k / {\simeq}$ (i.e. local symmetry equivalence classes) as follows:

\begin{proposition} \label{prop:qftype-sym}
    Let $\TT$ and $\TT'$ be $\cV$-structures. For any tuples $\vec{v} \in \TT^k$ and $\vec{u} \in {\TT'}^k$, we have $\vec{v} \simeq \vec{u}$ if and only if $\qftype(\TT, \vec{v}) = \qftype(\TT', \vec{u})$. 
\end{proposition}
\begin{example} \label{exm:1-types}
    Consider a simplification of the example in \Cref{fig:pip}, where the top node and associated statements are omitted, and the beginning and ending nodes are exposed as constant symbols instead. There are 6 quantifier-free $1$-types that are realized in the topology family $\{ \TT_n \}_{n \ge 3}$: 
\begin{center}
\includegraphics[width=\textwidth,alt={The 6 quantifier-free 1-types that are realized in the simplified pipeline topology family with instance size at least 3.}]{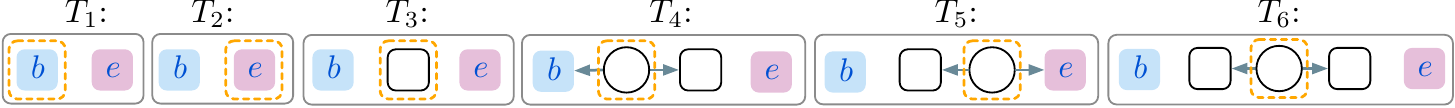}
\end{center}
where in each case, the tuple (of size 1) realizing the $1$-type is marked. \qed
\end{example}


Under certain finiteness conditions, a \emph{single} quantifier-free formula is sufficient to determine a quantifier-free $k$-type. We say a family $\cT$ of topologies is \emph{uniformly locally finite} (see \cite[Sec.~6.4]{Hodges1997}) if for any $d \in \NN_+$, there is a number $n \in \NN$ such that the size of any neighbourhood generated by $d$ nodes is bounded above by $n$, in any member of $\cT$. Practically, such topologies can encode networks where every process is connected to a bounded number of shared resources through bounded set of edge labels. Under the assumption of a finite vocabulary and uniform local finiteness, we can establish two important facts:
\begin{lemma} \label{lem:qftype-fin}
Given a topology family that is uniformly locally finite and has a finite vocabulary, we have:
\begin{description}[font=\normalfont]
\item [(L1)] There are finitely many quantifier-free types (or equivalently local symmetry equivalence classes), and each is definable by a single quantifier-free formula. Moreover, these formulas are computable.  
\item [(L2)] Every quantifier-free formula is equivalent to the disjunction of finitely many quantifier-free formulas, where each defines a distinct local symmetry equivalence class (as per the above item).   
\end{description}
\end{lemma}


\begin{example}[Continuation of \Cref{exm:1-types}]
Assume every {\tt pos} is initialized to {\tt true}. Suppose the vocabulary $\VNode$ consists of unary function symbols $\mathtt{l}$ and $\mathtt{r}$ (for left and right), constant symbols $\mathtt{b}$ and $\mathtt{e}$, and the unary predicate symbol $\mathtt{isProc}$ indicating a process/circle node. Then, the quantifier-free formula $\alpha_1[\nv]$ defining $T_1$ from \Cref{exm:1-types} is 
\[(\nv = \mathtt{l}(\nv) = \mathtt{r}(\nv) = \mathtt{b}) \land (\nv \neq \mathtt{e}) \land (\mathtt{e} = \mathtt{l}(\mathtt{e}) = \mathtt{r}(\mathtt{e})) \land \neg \mathtt{isProc}(\nv) \land \neg \mathtt{isProc}(\mathtt{e}).\] 
\end{example}

The lemma is analogous to \cite[Theorem 3.15]{Libkin2004}, but the proof is slightly more involved due to the function symbols in $\cV$. Intuitively, uniform local finiteness ensures that it suffices to consider quantifier-free formulas where the terms are of a bounded height; combined with the finiteness of the vocabulary, it follows that there are finitely many such formulas.

Given an Ashcroft assertion $\forall \vec{\nv}.\, \varphi[\vec{\nv}]$, we \emph{normalize} as follows:
\begin{enumerate}
    \item By pushing negations inward, $\varphi[\vec{\nv}]$ is rewritten as a positive Boolean combination of formulas of two distinct forms $\phi_\Node[\vec{\nv}]$ and $\phi_\Data[\vec{\nv}]$, where $\phi_\Node$ is a quantifier-free $\VNode$-formula. 
    
    \item Using Lemma \ref{lem:qftype-fin}(L2), we replace each sub-formula of the form $\phi_\Node[\vec{\nv}]$ with a disjunction of $\alpha_r$'s, where $\alpha_1, \ldots, \alpha_m$ define the $\simeq$-equivalence classes in the topology family. This results in a formula $\varphi'[\vec{\nv}]$.

    \item We convert $\varphi'[\vec{\nv}]$ into disjunctive normal form. Observe that,  since $\alpha_r$ and $\alpha_{r'}$ ($r \neq r'$) define two disjoint equivalence classes, we have $\alpha_r \land \alpha_{r'} \equiv \bot$. Hence, each disjunct can have at most one appearance of some $\alpha_r$. To ensure that each disjunct has precisely one appearance of some $\alpha_r$, we do this process instead to the equivalent formula $\varphi'[\vec{\nv}] \land \bigvee_{r=1}^m \alpha_r$.


    \item After the previous step, we end up with a formula like $\bigvee_{r=1}^m \alpha_r[\vec{\nv}] \land \varphi'_r[\vec{\nv}]$. But, the $\VNode$-terms that appear under $\varphi'_r$ can be arbitrary. This means that the same node can be referenced in several different ways. As the last step, we put these $\VNode$-terms into normal forms so that we impose a unique way of referencing each node. Fortunately, one can argue (see Remark \ref{rem:1}) that for every $\simeq$-equivalence class, there is a bounded set of representative terms such that every node can be referenced in a unique way. Let $t_r^{(1)}, \ldots, t_r^{(n_r)}$ be a list of representative terms. We change all references to nodes in $\psi_r[\vec{\nv}]$ to use the representative terms and get the following formula as the result of the normalization process: $ \forall \vec{\nv}.\, \bigvee_{r=1}^m \alpha_r[\vec{\nv}] \land \phi_r[\sval(t_r^{(1)}[\vec{\nv}]), \ldots, \sval(t_r^{(n_r)}[\vec{\nv}])]$. 
\end{enumerate}

\begin{remark}\label{rem:1}
For any $k$-tuple $\vec{v}$ over $\TT \in \cT$, the neighbourhood $\nb(\vec{v})$ is of bounded size. Let the tuple $\vec{u}$ represent all nodes in $\nb(\vec{v})$ and for each $i = 1, \ldots|\nb(\vec{v})|$. A set of representative terms chooses a term for each $u_i$, which are by construction free of repetitions. Alternatively, equalities between different possible choices of representative terms appear as formulas in $\qftype(\TT, \vec{v})$, and the choice of the representative term can be viewed as choosing one of the equivalent terms. We call $\vec{t}$ a list of \emph{representative terms} for $\qftype(\TT, \vec{v})$ where $t^{(i)}(\vec{v}) = u_i$.
%
\end{remark}

\begin{theorem} \label{thm:nf}
    Under the condition in \Cref{lem:qftype-fin}, every Ashcroft assertion is equivalent over $\cT$ to a normalized Ashcroft assertion of the same width.
\end{theorem}

\begin{example}[Continuation of \Cref{exm:1-types}] \label{exm:nf}
    Let $\alpha_1[\nv], \ldots, \alpha_6[\nv]$ be quantifier-free formulas defining the types $T_1, \ldots, T_6$. Consider the given Ashcroft assertion
    \[
        \forall \nv. (\mathtt{isProc}(\nv) \land \mathtt{l}(\nv) \neq \mathtt{b}) \to \texttt{ready}(\sval(\mathtt{r}(\nv))) \to \texttt{data}(\sval(\mathtt{b})) < \texttt{data}(\sval(\mathtt{r}(\nv))).
    \]

After step 1, a $\phi_\Node$ subformula in this can be $\neg (\mathtt{isProc}(\nv) \land \mathtt{l}(\nv) \neq \mathtt{b})$. It is satisfied in the topology family by a union of equivalence classes, which in this case correspond to $T_1, \dots, T_4$, and hence this subformula is replaced in step 2 with $\bigvee_{r=1}^4 \alpha_r[\nv]$.
After step 3, the formula $\varphi'_5[\nv]$, corresponding to $T_5$ (from \Cref{exm:1-types}), is 
 $ \texttt{ready}(\sval(\mathtt{r}(\nv))) \to \texttt{data}(\sval(\mathtt{b})) < \texttt{data}(\sval(\mathtt{r}(\nv)))$. 
If we choose $(\nv, \mathtt{l}(\nv), \mathtt{b}, \mathtt{e})$ as the representative terms for $T_5$, then in step 4 this is rewritten to $\texttt{ready}(\sval(\mathtt{e})) \to \texttt{data}(\sval(\mathtt{b})) < \texttt{data}(\sval(\mathtt{e}))$. 
However, if we choose $(\nv, \mathtt{l}(\nv), \mathtt{b}, \mathtt{r}(\nv))$ as the representative terms, then $\varphi'_5[\nv]$ stays as it is.
\qed
 \end{example}

%
%

\subsection{Hoare Triples From a Normalized Ashcroft Invariant}\label{sec:ht}
A normalized Ashcroft assertion is basically a template invariant of width $k$. To instantiate the template, we need to search for interpretations for its uninterpreted symbols. Given a specific program, the verification conditions for the validity of the invariant can be used as constraints in this search. Below, we define precisely what these verification conditions are in the form of Hoare triples obtained from the normalized Ashcroft invariant for a specific program. 

\begin{wrapfigure}[10]{r}{0.5\textwidth}\vspace{-14pt}
\begin{center}
\fbox{
\begin{minipage}{0.47\textwidth}
$\forall \vec{w} \in \TT^k, \forall v_0, w \in \TT$:\vspace{-3pt}
\begin{align}
\{\bigwedge_{v \in \TT} \init(v)\} &\ {\epsilon}\ \{\varphi[\vec{w}]\}  \label{htri:init} \tag{\sc Init}\\  
\{\bigwedge_{\vec{v} \in \TT^k} \varphi[\vec{v}]\} &\ {v_0}\ \{\varphi[\vec{w}]\} \label{htri:cont} \tag{\sc Cont}\\  
\{\bigwedge_{\vec{v} \in \TT^k} \varphi[\vec{v}]\} &\ {\epsilon}\ \{\neg \error(w)\} \label{htri:err} \tag{\sc Safe}
\end{align}
\end{minipage}}\vspace{-7pt}
\caption{Set $H_\Phi^{\TT}$ of Hoare triples for $\Phi$ \label{fig:hts}}
\end{center}
\end{wrapfigure}
Fix vocabularies $\VNode$ and $\VData$. For any formula $\varphi[\nv_1, \ldots, \nv_k]$ in \QFAshcroft{}, topology $\TT$, and tuple $\vec{v}$ in $\TT^k$, expanding the vocabulary $\VNode$ with parameters from $\TT$, we can obtain a closed \QFAshcroft{} formula $\varphi[\vec{v}]$ by substituting $\vec{v}$ for the free variables. Then, for a given topology $\TT$, any Ashcroft assertion $\Phi := \forall \vec{\nv}.\, \varphi[\vec{\nv}]$ is equivalent to a conjunction $\bigwedge_{\vec{v} \in \TT^k} \varphi[\vec{v}]$, and we can relate $\Phi$ to a set of Hoare triples in Figure \ref{fig:hts}.

In \cite[Proposition 6.3]{pps}, valid Hoare triples are obtained from an Ashcroft invariant in the same style. We observe that the \emph{converse} also holds.
\begin{proposition}\label{prop:ashcroft-htri}
    Let $\prog$ be a program over a finite topology $\TT$. For any $\Phi$, the Hoare triples in $H_{\Phi}^{\TT}$ are valid if \textbf{and only if} $\Phi$ is an Ashcroft invariant for $\prog$.
\end{proposition}

We say a Hoare triple is in normal form if its precondition and postcondition are $\Assertion{\TT}$ formulas for some $\VNode$-structure $\TT$.
\begin{proposition} \label{prop:htri-nf}
    Any Hoare triple in $H_{\Phi}^{\TT}$ is equivalent to one in normal form.
\end{proposition}
In particular, for a normalized Ashcroft assertion $\forall \vec{\nv}.\, \varphi[\vec{\nv}]$, we have
\[ \varphi[\vec{v}] \equiv \phi_{r}[\sval(t_{r}^{(1)}(\vec{v})), \ldots, \sval(t_{r}^{(n_{r})}(\vec{v}))] \ \ \ \ \ (1\le r \le m) \]
where $r$ is the unique number such that $T_r = \qftype(\TT, \vec{v})$, also denoted by $r(\vec{v})$.


For every node $v$, let $x_v$ be a fresh $\Data$-sorted variable; it is fine to overlap names between different topologies. We define a substitution operation on any $\Assertion{\TT}$ or $\Transition{\TT}$ formula $\varphi$: let $\subst(\varphi)$ be the $\VData$-formula obtained from $\varphi$ by substituting $\sval(v)$ with $x_v$ and $\sval'(v)$ with $x'_v$ for all $v \in \TT$. In addition, we use $\varphi'$ to denote the formula obtained from $\varphi$ by replacing every $\sval$ with $\sval'$. The following proposition relates the validity of a Hoare triple in normal form to an entailment in the background theory $\DD$ for program data. We define $\gform_{\TT}\sem{v}$ as the formula $\sem{v} \land \bigwedge_{u \in \TT \setminus \nb(v)} (\sval(u) = \sval'(u))$.



\begin{proposition} \label{prop:valid-entail}
    Let $\varphi$ and $\psi$ be $\Assertion{\TT}$ formulas.
    \begin{itemize}
        \item For any $a \in \TT$, the Hoare triple $\htri{\varphi}{a}{\psi}$ is valid if and only if $\DD \models \subst(\varphi \land \gform_{\TT}\sem{a} \Rightarrow \psi')$.
        \item The Hoare triple $\htri{\varphi}{\epsilon}{\psi}$ is valid if and only if $\DD \models \subst(\varphi \Rightarrow \psi)$.
    \end{itemize}
\end{proposition}


\subsection{CHC Encoding for Normalized Ashcroft Invariants}


We start by using the observations from Propositions \ref{prop:ashcroft-htri} and \ref{prop:valid-entail} to set up a search for an Ashcroft invariant for a single program in the topology as a CHC encoding. From the normalization procedure (\Cref{thm:nf}), we collect the quantifier-free $k$-types $T_r$ and their corresponding representative terms $t^{(1)}_r, \ldots, t^{(n_r)}_r$, for each $T_r$ ($ 1 \le r \le m$). Let $\prog_i$ be a program over a topology $\TT_i \in \cT$. The CHC system for searching for an Ashcroft invariant of width $k$ for $\prog_i$ is based on unknown predicates of the form $\Inv_i : \DD^{n_i} \to \bool$ ($1 \le i \le m$). 

The clauses are listed in Figure \ref{fig:chc}. For any $\TT \in \cT$ and tuple $\vec{v} \in \TT^k$, the tuple $\vec{x}_{\vec{v}}$ abbreviates $(x_{v_1}, \ldots, x_{v_k})$ and the tuple $U(\vec{v})$ lists the neighbours of $\vec{v}$ in the same order as $t_{r(\vec{v})}$, i.e., $U(\vec{v}) = (t_j^{(1)}(\vec{v}), \ldots, t_j^{(n_j)}(\vec{v}))$, where $j = r(\vec{v})$.

The clauses are based on \Cref{prop:valid-entail} for the validity of the Hoare triples $H_{\Phi}^{\TT_i}$ (see \Cref{fig:hts}), which certify that the Ashcroft assertion $\Phi$ is an invariant for the program over $\TT_i$.


\begin{figure}[t]
\begin{center}
\fbox{
\begin{minipage}{\textwidth}
$\forall \vec{w} \in \TT^k, \forall v_0, w \in \TT_i$:\vspace{-15pt}
\begin{align}
    \bigwedge_{v \in \TT_i} \subst(\init_i(v)) &\implies \Inv_{r(\vec{w})}(\vec{x}_{U(\vec{w})}) \label{chc:init} \\
    \bigwedge_{\vec{v} \in \TT_i^k} \Inv_{r(\vec{v})}(\vec{x}_{U(\vec{v})}) \land \subst(\gform_{\TT_i}\sem{ v_0 }_i) &\implies 
    \Inv_{r(\vec{w})}(\vec{x}'_{U(\vec{w})}) \label{chc:cont} \\
    \bigwedge_{\vec{v} \in \TT_i^k} \Inv_{r(\vec{v})}(\vec{x}_{U(\vec{v})}) &\implies \subst(\neg \error_i(w)) \label{chc:safe}
\end{align}
\end{minipage}}\vspace{-3pt}
\caption{Constrained Horn Clauses for a single program over topology $\TT_i$ \label{fig:chc}}
\end{center}
\end{figure}

\begin{theorem} \label{thm:chc-one}
    The CHC system of Figure \ref{fig:chc} is (syntactically) satisfiable iff an Ashcroft invariant of width $k$ for $\tp{\TT_i, \sem{-}_i}$ w.r.t. $\tp{\init_i, \error_i}$ exists. Moreover, the invariant has the form $\forall \vec{\nv}.\, \bigvee_{r=1}^m \alpha_r[\vec{\nv}] \land \Inv_r[\sval(t_r^{(1)}[\vec{\nv}]), \ldots, \sval(t_r^{(n_r)}[\vec{\nv}])]$.  
\end{theorem}


%
%

\section{Compositional Proof Search for an Ashcroft Invariant}\label{sec:encoding}
In Section \ref{sec:chc-single}, we established how to effectively search for ingredients of an Ashcroft invariant for a single program. Combining this with the observation of Theorem \ref{thm:ashcroft-basis}, one can immediately see how this can be extended to an entire parameterized program. Intuitively, we list all the constraints for all the programs in the basis as stated in Theorem \ref{thm:ashcroft-basis}. 

\begin{theorem} \label{thm:chc-param}
    Let $\cP = \{ \tp{\TT_i, \sem{-}_i} \}_{i\in\cI}$ be a parameterized program whose topologies satisfy the finiteness assumptions. Let $\cS = \{ \tp{\init_i, \error_i} \}_{i \in \cI}$ be a parameterized specification. For any $k \in \NN$, given a finite subset of $\cT$ that constitutes a basis of rank $k + 1$, one can effectively construct a CHC system that has a solution if and only if there is an Ashcroft invariant of width $k$ for $\cP$ w.r.t. $\cS$.
\end{theorem}

There are however two problems with this process as a practical algorithm: (1) Practically no user-friendly topology family has a finite basis to satisfy the conditions of Theorem \ref{thm:ashcroft-basis}, and (2) Even for those that do, this process may produce a CHC system with a very high amount of {\em redundancy}.  In \cite{pps}, the first problem is addressed through a procedure that augments topologies with all the missing small programs to form a finite basis. 
We use this idea and turn it into a reasonable algorithm that permits the user to specify the input parameterized program using a natural family of topologies (e.g.  {\em rings}) that does not contain a finite basis, and generates the CHC encoding by efficiently deciding what small programs to add and which constraints from Figure \ref{fig:chc} to include.  

\subsection{Subprograms and Downward Closure}\label{sec:closure}

    Given a program $\prog = \tp{\TT, \sem{-}}$, for any substructure $\SS$ of $\TT$, the \emph{subprogram} of $\prog$ over $\SS$ is the program $\prog|_{\SS} = \tp{\SS, \sem{-}|_{\SS}}$.
Note that a safety specification for $\prog$ can be restricted to a safety specification for any subprogram. Adding subprograms to a parameterized program always preserves unsafety. We show that under certain conditions, it also preserves safety. Given a topology $\TT$, we say an initial condition $\init$ is \emph{extensible} if for any substructure $\SS$ of $\TT$, for any initial global state $\val$ on $\SS$, there is an initial global state $\val'$ on $\TT$ such that $\val' |_{\SS} = \val$.

\begin{theorem} \label{thm:subprogram-safety}
    Let $\cP$ and $\cP'$ be parameterized programs such that $\cP \subset \cP'$ and every program of $\cP'$ is a subprogram of some program in $\cP$. Given any parameterized safety specification for $\cP$, the safety of $\cP'$ implies the safety of $\cP$, and the converse holds if the initial condition on every topology in $\cP$ is extensible.
\end{theorem}

For any topology $\TT$, define $\substructures_k(\TT)$ as the set of substructures of $\TT$ generated by $k$ not necessarily distinct nodes. For any parameterized program $\cP = \{\prog_i\}_{i\in\cI}$, where $\prog_i$ is over topology $\TT_i$, define $\clsk{\cP}{k}$ as the parameterized program obtained from $\cP$ by adding the subprogram $\prog_i|_{\SS}$ for any $i \in \cI$ and $\SS \in \substructures_k(\TT_i)$. The following proposition, which follows from \cite[Proposition 6.7]{pps}, presents a finite basis of rank $k$ contained in the topologies of $\clsk{\cP}{k}$ under the assumptions in \Cref{thm:nf}.
\begin{proposition} \label{prop:topo-sub}
    Let $k \in \NN_+$ and $\cT$ be a locally uniformly finite class of topologies over a finite vocabulary. Let $\cS := \bigcup_{\TT \in \cT} \substructures_k(\TT)$. Then $\cS$ is a basis for $\cT$ of rank $k$ that is finite up to isomorphism.
\end{proposition}

\subsection{CHCs for Downward Closure}\label{sec:algorithm}

\begin{algorithm}[t]
\caption{CHCs for a Downward Closure}
\label{alg:cl}
\begin{algorithmic}[1]
\Require Program $\cP = \{ \tp{\TT_i, \sem{-}_i} \}_{i\in\cI}$ on topologies $\cT = \{ \TT_i \}_{i\in\cI}$ satisfying the finiteness assumptions, $\cS = \{ \spec_i \}_{i\in\cI}$ is a parameterized specification, and $k \in \NN_+$
\Procedure{Chcs-for-Downward-Closure}{$\cP, \cS, k$}

\State compute $\alpha_1, \ldots, \alpha_m$ and representative terms $t_r^{(j)}$ as per \Cref{thm:nf}
\State $\cC \gets \varnothing$

\ForAll{$(\TT_i, \vec{w}) \in \Call{QF-Types}{\cT, k}$}
    \State $\cC \gets \cC \cup \left\{\bigwedge_{v \in \nb(\vec{w})} \subst(\init_i(v)) \Rightarrow \Inv_{r(\vec{w})}(\vec{x}_{U(\vec{w})})\right\}$
\EndFor

\ForAll{$(\TT_i, v_0\,\vec{w}) \in \Call{QF-Types}{\cT, k + 1}$}
    \State $\phi \gets \subst(\gform_{\nb(v_0\vec{w})}\sem{ v_0 }_i)$
    \State $\cC \gets \cC \cup \left\{\bigwedge_{\vec{v} \in \nb(v_0\vec{w})^k} \Inv_{r(\vec{v})}(\vec{x}_{U(\vec{v})}) \land \phi \Rightarrow 
    \Inv_{r(\vec{w})}(\vec{x}'_{U(\vec{w})})\right\}$
\EndFor

\ForAll{$(\TT_i, w) \in \Call{QF-Types}{\cT, 1}$}
    \State $\cC \gets \cC \cup \left\{\bigwedge_{\vec{v} \in \nb(w)^k} \Inv_{r(\vec{v})}(\vec{x}_{U(\vec{v})}) \Rightarrow \subst(\neg \error_i(w))\right\}$
\EndFor

\State \Return $\cC$
\EndProcedure
\end{algorithmic}
\end{algorithm}

If we blindly collect Hoare triples from all small programs in the basis, we end up with redundant Hoare triples with the same postcondition and command, but different preconditions, where one is subsumed by the other in a purely syntactic way. Avoiding these redundancies by constructing the full set and pruning it is clearly not a desirable solution. In \Cref{alg:cl}, we present an alternative. It takes as input a parameterized program $\cP$, a parameterized specification $\cS$, and a number $k \in \NN_+$, and returns a CHC system that encodes the existence of an Ashcroft invariant for the downward closure of $\cP$. The indexed sets $\cP$ and $\cS$ are given as computable functions. We assume there is an oracle $\Call{QF-Types}{\cT, k}$ that returns a list of all quantifier-free $k$-types in $\cT$, where each type is represented by a $k$-tuple in a topology in $\cT$. The algorithm may simplify the verification condition of one program by deferring to that of a smaller program, leveraging the fact that the \emph{smallest} substructure containing the nodes in a command and postcondition is in the downward closure.

\begin{theorem}
    The algorithm $\Call{Chcs-for-Downward-Closure}{\cP, \cS, k}$ returns a CHC system that is (syntactically) satisfiable if and only if there is an Ashcroft invariant of width $k$ for $\clsk{\cP}{k + 1}$ w.r.t. $\cS$.
\end{theorem}

There are $|\Call{QF-Types}{\cT, k}|$ unknown predicates produced by \Cref{alg:cl}, the same as \Cref{thm:chc-param}, and $O(|\Call{QF-Types}{\cT, k + 1}|)$ clauses are generated. Both are characteristics of the topology family, which reflect its complexity. For a specific topology family $\cT$, the procedure $\Call{QF-Types}{\cT, k}$ can often be implemented by enumerating node tuples of length $k$ in small instances of $\cT$. For instance, if $\cT$ is the family of lines in the style of \Cref{fig:pip} (a), it suffices to enumerate node tuples in line instances with up to $2 k + 1$ circle nodes.

\section{Complete Encoding Optimizations} \label{sec:opt}
A key feature of our solution is that the normalization process grants us flexibility and control in devising the {\em complete} template for proof search. In this section, we take advantage of this control to introduce a couple of {\em complete} optimizations for the encoding. We first lay the formal foundation for these optimizations through a {\em parametric} variance of the normalization process introduced in \Cref{sec:nai}. Concrete optimizations are defined as instantiations of the parameter. This means that we can give a single proof of completeness for the template (\Cref{thm:nf2}) and the completeness of all individual optimizations presented in this section immediately follows. The soundness is trivial, as any instantiation of a restricted template can be converted into an instantiation of a full template.

\subsection{Parametric Normalization of Ashcroft Invariants} \label{sec:nf2}

Let us continue with the normalization procedure in \Cref{sec:nai}. We claim that with further syntactic manipulations to a normalized Ashcroft assertion, the formulas $\phi_r$'s can be guaranteed to be {\em symmetric}: Let the types in the topology family be $T_1, \ldots, T_m$ and $\vec{v}_r$ be a node tuple that realizes $T_r$ for each $r$. We say the formula $\phi_r$ is \emph{symmetric} (w.r.t. the type $T_r$) if for any automorphism $f$ on $\nb(\vec{v}_r)$, we have $\phi_r[x_1, \ldots, x_{n_r}] \equiv \phi_r[x_{\pi_f(1)}, \ldots, x_{\pi_f(n_r)}]$, where $\pi_f$ is the permutation on indices $\{ 1, \ldots, n_r \}$ induced by $f$ and the representative terms of $T_r$. We exploit the symmetry of these formulas in \Cref{sec:tmk}.

Moreover, some $\phi_r$'s can be assumed to be $\top$ (true).  
We use an {\em indicator} vector $\vec{\chi} \in \{\top, \bot\}^m$ to {\em indicate} which. We say $\vec{\chi}$ \emph{selects a basis of rank $k$} from $\cT$ if the collective neighbourhoods of the node tuples at the indicated indices form a basis of rank $k$ for the topology; formally, if $\cU := \{ \nb(\vec{v}_r) : \chi_r \equiv \top \}$ is a basis of rank $k$ for $\cT$.

%
%

\begin{theorem} \label{thm:nf2}

    Assume the setting of \Cref{thm:nf}. Let $\vec{\chi} \in \{\top,\bot\}^m$ be any indicator vector that selects a basis of rank $k$ from $\cT$. Then over $\cT$, any Ashcroft assertion is equivalent to a normalized Ashcroft assertion of the same width
\[ \forall \nv_1, \ldots, \nv_k.\, \bigvee_{r=1}^m \alpha_r[\vec{\nv}] \land \phi_r[\sval(t_r^{(1)}[\vec{\nv}]), \ldots, \sval(t_r^{(n_r)}[\vec{\nv}])] \]
    such that for every $r$,  $\phi_r$ is symmetric w.r.t. the type $T_r$, and $\phi_r \equiv \top$ if $\chi_r \equiv \bot$.
\end{theorem}

The proof idea is that, given a formula $p[x_1, x_2]$, the conjunction $p[x_1, x_2] \land p[x_2, x_1]$ is invariant under any permutation on the indices $\{1,2\}$.


We use the symmetry of $\phi_r$'s in \Cref{sec:tmk} and the indicator vector $\vec{\chi}$ to define optimizations in this section. Given $\vec{\chi}$, we may modify \Cref{alg:cl} by skipping $\vec{w}$ such that $\chi_{r(\vec{w})} \equiv \bot$ in line 4 and 6 and, similarly, skipping $\vec{v}$ such that $\chi_{r(\vec{v})} \equiv \bot$ in line 8 and 10. For an input parameterized program $\cP$ to the algorithm, this simplification is \emph{complete} if $\vec{\chi}$ selects a basis of rank $k$ from the topology family of $\clsk{\cP}{k + 1}$. 

\subsection{One Predicate per Neighbourhood} \label{sec:per-nbhd}
\def\opn{{\tt OPN}}

By default, our CHC encoding uses one unknown predicate per quantifier-free $k$-type $T_1, \ldots, T_m$ in $\cT$. Let $\UU_r$ be the neighbourhood of an arbitrary tuple with type $T_r$, which is unique up to isomorphism. We define an equivalence relation on the types $T_1, \ldots, T_m$ based on these neighbourhoods: we say $T_r \sim T_s$ if $\UU_r$ is isomorphic to $\UU_s$. Then, the $\sim$ relation provides a choice of $\vec{\chi}$ that always induces a complete optimization to \Cref{alg:cl}: for each equivalence class in $\{ T_1, \ldots, T_m \}/{\sim}$, $\chi_r$ chooses precisely one representative from the class. The net effect is that we use one unknown predicate per neighbourhood generated by $k$ nodes. In \Cref{sec:experiments}, we refer to this optimization as \opn.

\subsection{Distinct Process Nodes as Generators} \label{sec:distinct}


We start by assuming that the nodes in a topology $\TT$ are partitioned into two kinds: process nodes and resource nodes; for example, circles and squares in Figure \ref{fig:pip}. Resource nodes do not execute code. We add a unary predicate $\isProc$ to the vocabulary $\VNode$ to indicate whether a node is a process node. We require that the neighbourhood of any resource node contains no process node, and every resource node is in the neighbourhood of some process node.

\def\dpg{{\tt DPG}}

\paragraph{Skipping Resource Nodes}

Recall that \Cref{alg:cl} implicitly adds subprograms over the topologies $\{ \nb(\vec{v}) \mid \vec{v} \text{ is a $(k+1)$-tuple over some } \TT \in \cT \}$ to guarantee a basis of rank $k + 1$. Under the modelling assumptions mentioned above, however, it suffices to consider $\vec{v}$ consisting of process nodes. We can exclude the rest, which are theoretically unnecessary by, for each $r$, setting $\chi_r$ to $\bot$ if $\vec{v}_r$ (an arbitrary tuple that realizes $T_r$, as defined in \Cref{sec:nf2}) contains a resource node. Meanwhile, in lines 6 and 9 of \Cref{alg:cl}, we also skip any $v_0$ and $w$ that are resource nodes.

\paragraph{Distinct Generators}

We can further exclude topologies generated by fewer than $k$ process nodes from the downward closure,  by setting $\chi_r$ to $\bot$ for any $r$ such that $\vec{v}_r$ contains repetitive nodes. Since line 9--10 of \Cref{alg:cl} implicitly closes $\cT$ under substructures generated by one node, we change them to the following:
\begin{algorithmic}
\ForAll{$(\TT, w) \in \Call{Substructures}{\cT, k, 1}$}
    \State $\cC \gets \cC \cup \left\{\bigwedge_{\vec{v} \in \TT^k :\, \chi_{r(\vec{v})}} \Inv_{r(\vec{v})}(\vec{x}_{U(\vec{v})}) \Rightarrow \subst(\neg \error_i(w))\right\}$
\EndFor
\end{algorithmic}
where $\Call{Substructures}{\cT, k, 1}$ returns, up to isomorphism, a list of all structures $\TT$ together with a distinguished node $w \in \TT$ such that $\TT$ is generated by $k$ distinct process nodes over any topology in $\cT$. We refer to the combined optimization in this subsection as \dpg.

\section{Theoretical Case Study: Deriving the Solution of \cite{HoenickeMP2017}} \label{sec:tmk}
To illustrate that the solution presented in this paper is robust, we look at the specific classic topology of {\em identical replicated threads sharing boundedly many global variables}, which is the premise of \cite{HoenickeMP2017}, and show that the solution of \cite{HoenickeMP2017} is a special case of our solution.

\begin{wrapfigure}[5]{r}{0.15\textwidth}\vspace{-12pt}
\includegraphics[scale=0.18,alt={The star topology, where $n$ identical processes are connected to a single resource node.}]{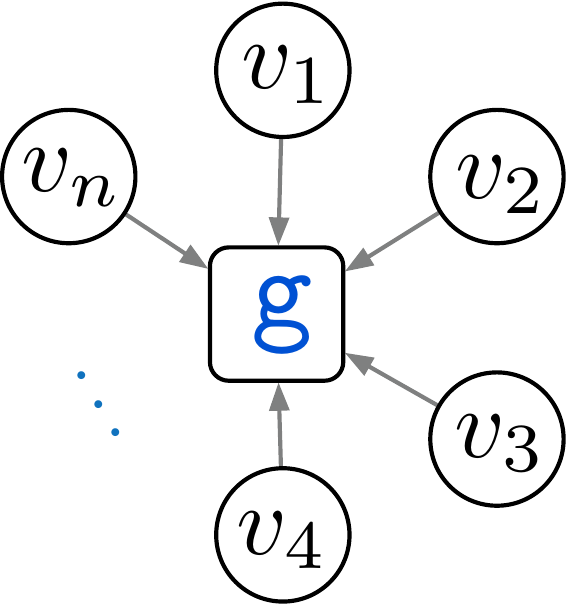}
\end{wrapfigure}
First, we model the topology as a {\em star}. The vocabulary $\VNode$ consists of a single constant symbol $\mathtt{g}$. The $\VNode$-structure $\TT_n$ has underlying set $\{0, 1, \ldots, n\}$, where $0$ is the interpretation of $\mathtt{g}$ (the resource node hosting all shared variables), and the positive numbers stand for the $n$ indistinguishable process nodes $v_i$'s. 

In \cite{HoenickeMP2017}, the safety specification is given as a pair of $\VData$-formulas $\id{Init}$ and $\id{Error}$, and the transitions are given by a $\VData$-formula $s$. To simplify the presentation, we fix $k = 2$ and assume $\id{Error}$ is stated for a single node. In our notation, for any $n \in \NN_+$ and $v \neq 0$, we have $\init_n(v) = \id{Init}(\sval(0), \sval(v))$, $\error_n(v) = \id{Error}(\sval(0), \sval(v))$, and $\sem{v}_n = s(\sval(0), \sval(v), \sval'(0), \sval'(v))$; for $v = 0$, we have $\init_n(v) = \top$, $\init_n(v) = \bot$, and $\sem{v}_n = \bot$.

Let $\cP := \{ \tp{\TT_n, \sem{-}_n} \}_{n \ge 2}$. The refined downward closure process in \Cref{sec:distinct} does not add any new subprogram to $\cP$, as the topology family already contains $\{ \TT_2, \TT_3 \}$, which is also a basis of rank $3$. 

Our normalization produces   
    $\forall u_1, u_2.\, \alpha[u_1, u_2] \to \phi[\sval(\mathtt{g}), \sval(u_1), \sval(u_2)]$,
where $\alpha[u_1, u_2] := u_1 \neq u_2 \land \isProc(u_1) \land \isProc(u_2)$, as a  \emph{complete} template for an Ashcroft assertion of width $2$ for $\cP$, which incidentally is precisely how an Ashcroft assertion is defined in \cite{HoenickeMP2017}.


Running \Cref{alg:cl} modified with the optimization described in \Cref{sec:distinct}, with minor syntactic simplification, we obtain the clauses $\cC$ in Figure \ref{fig:tmk}.
\begin{figure}[t]
\begin{center}
\begin{minipage}{\textwidth}
\begin{align*}
\id{Init}(x_0, x_1) \land \id{Init}(x_0, x_2) &\implies \Inv(x_0, x_1, x_2) \\
\Inv(x_0, x_1, x_2) \land {\color{gray} \Inv(x_0, x_2, x_1)} \land s(x_0, x_1, x'_0, x'_1) &\implies \Inv(x'_0, x'_1, x_2) \\
\Inv(x_0, x_1, x_2) \land {\color{gray} \Inv(x_0, x_2, x_1)} \land s(x_0, x_2, x'_0, x'_2) &\implies \Inv(x'_0, x_1, x'_2) \\
\Inv(x_0, x_1, x_2) \land \Inv(x_0, x_3, x_2) \land \Inv(x_0, x_1, x_3) & \\
{\color{gray} \Inv(x_0, x_2, x_1) \land \Inv(x_0, x_2, x_3) \land \Inv(x_0, x_3, x_1)}
{}\land s(x_0, x_3, x'_0, x'_3) &\implies \Inv(x'_0, x_1, x_2) \\
\Inv(x_0, x_1, x_2) \land {\color{gray} \Inv(x_0, x_2, x_1)} &\implies \neg \id{Error}(x_0, x_1)
\end{align*}
\end{minipage}\vspace{-5pt}
\caption{Clauses generated by Algorithm \ref{alg:cl} for the star topology.  \label{fig:tmk}}\vspace{-10pt}
\end{center}
\end{figure}

By the discussion in \Cref{sec:nf2}, without loss of generality, we may assume that $\Inv(x_0, x_1, x_2) \equiv \Inv(x_0, x_2, x_1)$ for all $x_0, x_1, x_2$, which prunes the clauses illustrated in {\color{gray} gray}, and results in a set of clauses $\cC'$ that is equisatisfiable to $\cC$.
The set of clauses $\cC'$ is exactly the CHC encoding for $2$-thread-modular proofs \cite[Figure 6]{HoenickeMP2017}, except that we do not treat control locations explicitly.

\section{Experimental Results}\label{sec:experiments}

\def\cblue{\color[rgb]{0.918545,0.000000,0.172660}}
\def\cred{\color[rgb]{0.000000,0.317275,0.826623}}
\def\cgreen{\color[rgb]{0.327398,0.834854,0.162883}}


\paragraph{Implementation} 
Our tool, \tool, takes as input a number $k \in \NN_+$, a parameterized program over topologies $\cT = \{ \TT_i \}_{i\in\cI}$ over a finite vocabulary $\cV$ (assuming the indices $\cI$ are consecutive natural numbers), a parameterized safety specification, and a pair of numbers $i, i' \in \cI$ with the assumption that all quantifier-free $(k+1)$-types appear in $\cT$ already appear in $\{ \TT_j \mid i \le j \le i' \}$. For the parameterized safety specification, we allow a more general form of error states that can relate $m$ nodes in a topology, where $m \le k$. 
{\tool} implements \Cref{alg:cl} with optimization options {\opn} and {\dpg} (\Cref{sec:opt}). The output encoding is in SMT-LIB2 format and can be passed to a variety of existing CHC solvers \cite{KomuravelliGC2016,Eldarica,Golem}. Here, we report the results for Z3/Spacer \cite{Z3,KomuravelliGC2016} except for one benchmark, which can be solved by Golem \cite{Golem} but not Spacer. Details about the experiment can be found in \refapp{app:experiments}.

\paragraph{Evaluation Goals} The main goal of our evaluation is straightforward: Establishing whether {\tool} can prove interesting properties of parameterized programs over a rich family of topologies.  As a secondary goal, we evaluate the impact of the optimizations proposed in \Cref{sec:per-nbhd,sec:distinct}.

\paragraph{Benchmarks}
There is not a lot of diversity in the example pool in the literature when it comes to different topologies for parameterized programs.
Hence, beyond including a few examples from the literature, we designed many more benchmarks over many more topologies. In particular, we have 30 programs over (1) rings, (2) lines (like example of Figure \ref{fig:pip}), (3) grids, (4) bounded-depth trees~\cite{pps}, and (5) binary trees (full list in \refapp{app:benchmarks}). The generality and flexibility offered by logical structures are very helpful in modelling some extra requirements for specialized versions of some of these topologies. For instance, a direction predicate for the line and an orientation predicate for the ring can be added effortlessly, and our benchmark set also includes these variations.

\begin{figure}[t]
\begin{center}
\includegraphics[width=\textwidth,alt={(a) is the quantile plot for the }]{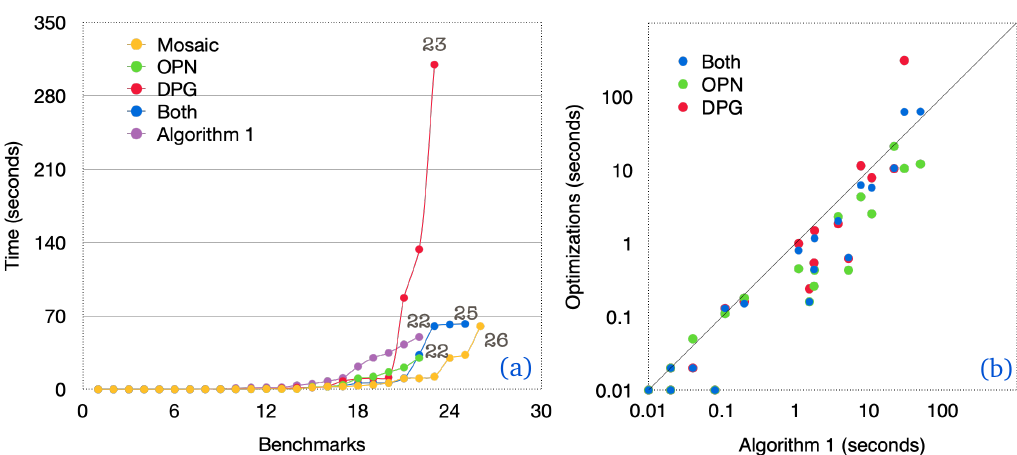}\vspace{-5pt}
\caption{(a) Quantile plot. (b) Scattered plot of optimizations vs baseline.}\vspace{-10pt}
\label{fig:plots}
\end{center}
\end{figure}

\paragraph{Results}
Figure \ref{fig:plots} presents our results, where {\tool} corresponds to the best performance of all options. The quantile plot includes the count of benchmarks solved by each option. Of the 26 benchmarks solved by \tool, the winning time of 16 is by \opn, 8 by \dpg, 13 by the combination, and 4 by Algorithm \ref{alg:cl} (counting ties).
The optimizations clearly have a significant impact. Remarkably, they do not consistently lead to an improved time: There are a few outliers above the middle line in Figure \ref{fig:plots}(b), which incidentally do not include the 3 cases where Algorithm 1 times out. This is because they target the size of the encoding, but a smaller encoding is not necessarily always easier to solve for the backend solver. \refapp{app:data} lists the full results, including encoding sizes.

The 4 benchmarks that {\tool} cannot solve are due to one of the following reasons: (1) non-linear reasoning required (1 tree benchmark), (2) encoding is too large (1 grid benchmark), and (3) CHC solver fails (1 line and 1 ring benchmark).
%

\paragraph{Limitations and Discussion}
The key limitations are related to the unscalability and unpredictability of the backend CHC solvers. Large encodings (e.g. a grid example) cannot be handled. A minor change in Figure \ref{fig:pip} example can flip the status of its solvability, without much change in the encoding size. 


Another bottleneck for scalability is that the basis of rank $k + 1$ becomes too large as $k$ increases, especially for a larger vocabulary $\VNode$. For instance, this happens for width-$2$ invariants for grid examples, since the distinctness of the four neighbours causes the basis to include $347$ small programs even with optimizations, and one of two such examples times out.

Finally, we came across benchmarks (not included) that do not admit universally quantified inductive invariants as proofs, where a richer class of invariants or ghost variables are required to handle them.

\section{Related Work}\label{sec:relwork}
The problem of verification of parameterized programs is generally known to be undecidable~\cite{AptK1986}. Our focus is on the {\em algorithmic} verification of infinite-state parameterized programs. As such, we do not survey the rich literature on finite-state parameterized programs \cite{AbdullaD2016} nor deductive verification techniques \cite{AshmoreGT2019,PadonMPS+2016}. 

\paragraph{Compositional Proofs of Parameterized Programs}
Compositional proofs have been a central theme in concurrency verification. We made a thorough comparison with the work in \cite{HoenickeMP2017} in Section \ref{sec:tmk}, which is singled out because it is sound and complete, albeit for a fixed topology. Similar results, but for finite-state programs, have appeared in \cite{Namjoshi2007}. However, general cut-off results like ours and the one in \cite{HoenickeMP2017} for the star topology were not given.

The work in \cite{NamjoshiT2012,NamjoshiT2016,NamjoshiT2018,AshmoreGT2019} is noteworthy for the common themes of the use of symmetry in compositional verification. In particular, in \cite{NamjoshiT2016}, the insight of identifying \emph{balanced} (therefore \emph{locally symmetric}) nodes in compositional verification is used to establish \emph{compositional cutoffs}.
The programs are finite-state, and the key result is the decidability of compositional model checking in polynomial time. Specifically, the problem is formally framed as deciding, for all programs in the parameterized family, the existence of interference-free local proofs for every process, which collectively become a global universal invariant {\em with one quantifier}. Our work generalizes this view by permitting universal quantification of any depth and generalizing compositional cutoffs to bases of any rank. However, the starting point of \cite{NamjoshiT2016} is local proofs, whereas we start with global invariants and look for a sound and complete reduction to local proofs.

\paragraph{Handling Multiple Communication Topologies}
Several topologies are studied for finite-state programs with multiple communication primitives in \cite{AbdullaHH2013}, where the decision procedure for each topology requires a bespoke representation of the configurations and a pre-order on them that induces subconfigurations for abstraction. Later in \cite{AminofKRS+2018}, parameterized finite-state programs over a rich set of topologies are studied, where several uniform (across topology classes) decidability and complexity results are given. The focus for decidability is on what they call {\em homogeneous} topologies, and none of the topologies in our benchmark set is homogeneous. Note that there are known undecidability results \cite{EmersonN1995} about rings (not {\em homogeneous}), which clarify why our algorithm cannot be a decision procedure. For infinite-state programs, a rich class of topologies are targeted in \cite{AshmoreGT2019} in the context of deductive verification. However, the invariant templates are predetermined by the user, and their small model property relates to checking concrete verification conditions rather than safety through Ashcroft invariants.


\paragraph{Parameterized Model Checking of Infinite-State Programs}
For algorithmic verification of infinite-state parameterized programs, there exist techniques that automatically generate universally quantified invariants \cite{EmmiMM2010,GurfinkelSM2016,HojjatRSW2014,MonniauxG2016}, all for a single topology. One can view {\em thread-modular proofs at many levels} \cite{HoenickeMP2017,FarzanKP2015} as the conclusion of this line of work by proposing proof systems that form a strict hierarchy and are relatively complete w.r.t. Ashcroft invariants. 

\section{Conclusion and Future Work}
On the conceptual side, this paper observes that under certain conditions about the families of topologies, a universally quantified invariant of an unbounded family of concurrent parameterized programs over these topologies can be verified through a bounded number of checks (Theorem \ref{thm:ashcroft-basis}), which are limited to entailments in the underlying data domain (Theorem \ref{thm:ashcroft-fin}). Hence, verification conditions do not require any axiomatization of the underlying topology family. On the algorithmic side, using the above observations, we propose a {\em complete} algorithm for discovering such an invariant of a given width $k$, if one exists. 

Our implementation uses CHC encodings (and solvers) as a point of convenience. The solvers do not yet scale to the level of handling the encodings of larger programs. The paper offers two {\em complete} optimizations to improve the encoding, but a more practical algorithm will have to rely on a combination of abstraction and decomposition techniques to make the invariant search tractable. We hope that our conceptual contribution in separation of reasoning about data and topologies will bring new opportunities for discovery on this path.  

Recall that Algorithm \ref{alg:cl} relies on an oracle that computes the quantifier-free types for a topology family. In the implementation, these are computed through brute-force enumeration relying on upper-bounds that are manually calculated for the topology class. The technical problem of necessary/sufficient conditions that guarantee efficient computability of quantifier-free types remains open. 

Universally quantified invariants are clearly not enough to capture every proof (or even every specification) of properties of interest in these contexts. It will be interesting to investigate if similar cut-off results can be proven for richer classes of invariants with quantifier alternations (e.g., see \cite{KoenigPIA2020}). Independently, there is some evidence that commutativity-based reductions can help simplify programs (on the star topology) enough for a universally quantified invariant to be sufficient in some cases \cite{FarzanKP2024}, and it will be interesting to incorporate similar techniques for arbitrary topology families.

\section*{Data-Availability Statement}

Proofs and experiment details can be found in the appendix. Our tool {\tool} and the benchmarks used in the evaluation (Section \ref{sec:experiments}) can be found at \cite{ChengF2026a}, where version 0.1.0 was submitted for artifact evaluation and includes the first 21 benchmarks listed in \refapp{app:benchmarks}.

\begin{acks}
The authors were supported by a Discovery Grant from Natural Sciences and Engineering Research Council of Canada.
\end{acks}

\section*{Disclosure of Interests}
The authors have no competing interests to declare that are relevant to the content of this article.

\bibliographystyle{ACM-Reference-Format}
\bibliography{biblio}

\crefalias{section}{appendix}
\crefalias{subsection}{appendix}

\iftoggle{ext}{%
\appendix
\section{A Basis of Rank 3 for the Downwardly Closed Pipeline Family}\label{app:pipeline}

For a topology family \(\cT\), let \(\clsk{\cT}{k} := \bigcup_{\TT \in \cT} \substructures_k(\TT)\) (see \Cref{sec:closure}). Let \(\LL_n\) be the pipeline with \(n\) process nodes, as depicted in \Cref{fig:pip}. A basis of rank \(3\) for the family \(\clsk{\{\LL_n\}_{n \ge 3}}{3}\) is illustrated in Figure \ref{fig:pipb}.

\begin{figure}[htbp]
\begin{center}
\includegraphics[width=\textwidth]{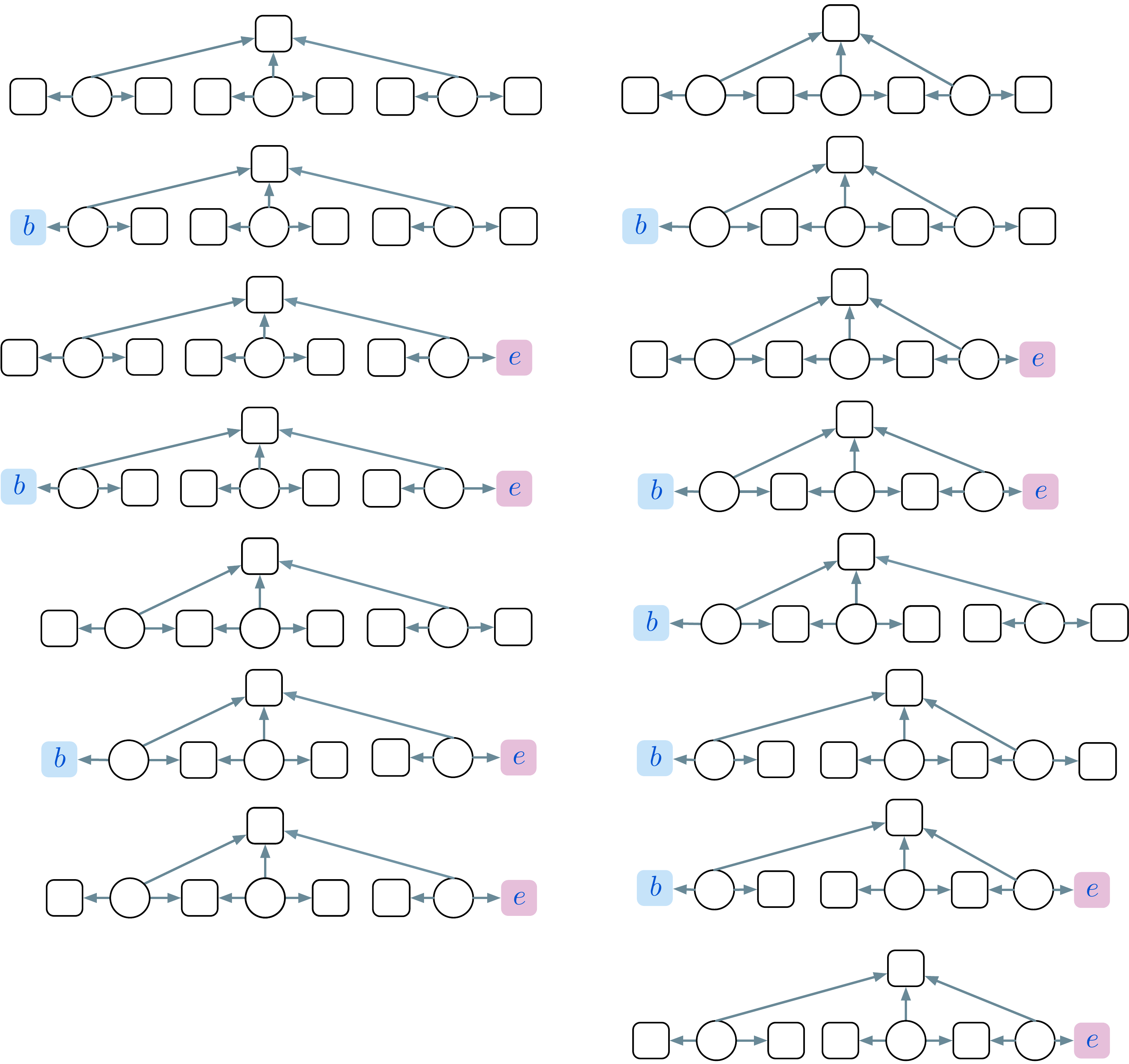}
\caption{The 15 topologies in a basis of rank 3 for \(\clsk{\{\LL_n\}_{n \ge 3}}{3}\)}
\label{fig:pipb}
\end{center}
\end{figure}

In general, we have the following theorem:
\begin{theorem}
  Any substructure with \(k\) generators of some \(\LL_n\) (\(n \in \NN_+\)) is an embedded substructure of \(\LL_m\) for some \(m \le 2k + 1\).
\end{theorem}
\begin{proof}[Proof Sketch]
  Let \(\SS\) be a substructure of \(\LL_n\) with \(k\) generators. Without loss of generality, we may assume the generators are process nodes. List them from left to right: \(v_{i_1}^{(n)}, \ldots, v_{i_k}^{(n)}\) where \(i_1 \le \cdots \le i_k\) (we use \(v_i^{(n)}\), where \(1 \le i \le n\), to denote the \(i\)-th process node in \(\LL_n\)). Define the sequence \(j_1, \ldots, j_k\) inductively as follows:
  \[
    j_t := \begin{cases}
      \min(i_1, 2), &\text{if } t = 1 \\
      j_{t-1} + \min(i_t - i_{t-1}, 2). &\text{otherwise}
    \end{cases}
  \]
  Let \(m := j_k + \min(n - i_k, 1)\). Then \(m \le 2k + 1\). One can verify that the substructure of \(\LL_m\) generated by \(v_{i_1}^{(m)}, \ldots, v_{i_k}^{(m)}\) is isomorphic to \(\SS\).
\end{proof}

\begin{corollary}
  Any quantifier-free \(k\)-type realized in \(\{\LL_n\}_{n \in \NN_+}\) is realized in \(\{\LL_m\}_{m \le 2k + 1}\).
\end{corollary}

\section{Generalized Ashcroft Invariants}\label{app:ash}

Formally, \emph{generalized Ashcroft assertions} are defined as a fragment of first-order formulas over a vocabulary $\cV_\Ashcroft = \VNode \uplus \VData \uplus \{ \sval \}$ of two sorts, $\Node$ and $\Data$, where the special symbol $\sval$ is typed $\Node \to \Data$. 
The fragment \QFAshcroft{} consists of first-order formulas over $\cV_\Ashcroft$ where no quantification is over $\Node$ and every free variable is of sort $\Node$. Ashcroft assertions of \emph{width} $k$, denoted by $\Ashcroft[k]$, are of the form $\forall u_1, \ldots, u_k.\, \varphi[u_1, \ldots, u_k]$, 
where $\varphi$ is a \QFAshcroft{} formula.  
The set of Ashcroft assertions \Ashcroft{} is the union of $\Ashcroft[k]$ for all $k \in \NN$.

We always fix the $\VData$-structure $\DD$. Therefore, the truth value of $\Phi$ is given by a specific $\VNode$-topology $\TT$ and a global state $\val$ on $\TT$, interpreting the symbol $\sval$. Given $\TT$, we define the satisfiability relation $\val \models \Phi$ and the entailment relation $\Phi \models \Phi'$ under $\TT$ in the standard way. Although both relations are parametric on $\TT$, since $\TT$ is clear from the context, we omit it from our notation.

For a program $\prog = \tp{\TT, \sem{-}}$ and a safety specification $\tp{\init, \error}$, a \emph{generalized Ashcroft invariant} is a formula $\Phi \in \Ashcroft$ such that for any global states $\sval$ and $\sval'$, \vspace{-3pt}
    \begin{description}
        \item[Initialization:] if $\val$ satisfies $\init$, then $\val \models \Phi$;
        \item[Continuation:] if $\val \models \Phi$ and $\val \to \val'$, then $\val' \models \Phi$;
        \item[Safety:] if $\val \models \Phi$, then $\val$ does not satisfy $\error$.
    \end{description}\vspace{-3pt}
$\Phi$ is an invariant for the parameterized program if it is an invariant for every member of the family.

\section{Key Proofs} \label{app:proofs}

\begin{proposition}[\Cref{prop:qftype-sym}]
    Let $\TT$ and $\TT'$ be $\cV$-structures. For any tuples $\vec{v} \in \TT^k$ and $\vec{u} \in {\TT'}^k$, we have $\vec{v} \simeq \vec{u}$ if and only if $\qftype(\TT, \vec{v}) = \qftype(\TT', \vec{u})$. 
\end{proposition}
\begin{proof}
    By induction over the structure of formulas, one can show that for any quantifier-free $\cV$-formula $\phi$ with $k$ free variables, tuple $\vec{v} \in \TT^k$, and local isomorphism $\beta$ from $\vec{v}$ to some tuple in $\TT'^k$, we have $\TT \models \phi(\vec{v})$ iff $\TT' \models \phi(\beta(\vec{v}))$.

    Let $\vec{v} \in \TT^k$ and $\vec{u} \in {\TT'}^k$. Suppose $\qftype(\TT, \vec{v}) = \qftype(\TT', \vec{u})$. Define a binary relation $\beta \subset \TT \times \TT'$ that consists of the pairs $(t(\vec{v}), t(\vec{u}))$ for all $\cV$-term $t$. It is easy to verify that $\beta$ is a local isomorphism from $\vec{v}$ to $\vec{u}$. For instance, to show well-definedness, suppose there are terms $t$ and $t'$ such that $t(\vec{v}) = t'(\vec{v})$. Then the formula $t = t'$ is in $\qftype(\TT, \vec{v})$ and thus $\qftype(\TT', \vec{u})$, which implies $t(\vec{u}) = t'(\vec{u})$.
\end{proof}

\begin{proposition}[\Cref{prop:htri-nf}]
    Any Hoare triple in $H_{\Phi}^{\TT}$ is equivalent to one in normal form.
\end{proposition}
\begin{proof}
    For any formula $\varphi[\nv_1, \ldots, \nv_k] \in \QFAshcroft$ and tuple $\vec{v} \in \TT^k$, the formula $\varphi[\vec{v}]$ is a Boolean combination of 
    formulas of the form $\phi_\Node[\vec{v}]$ and $\phi_\Data[\sval(t_1[\vec{v}]), \ldots, \sval(t_n[\vec{v}])]$ (see Step 1 of the normalization algorithm in \Cref{sec:nai}). The former is equivalent to either true or false. The latter is equivalent to the formula $\phi_\Data[\sval(u_1), \ldots, \sval(u_n)]$ where $\vec{u} = (t_1(\vec{v}), \ldots, t_n(\vec{v}))$. Hence, $\varphi[\vec{v}]$ is equivalent to a formula in $\Assertion{\TT}$, and so is any conjunction $\bigwedge_{\vec{v} \in \TT^k} \varphi[\vec{v}]$.
\end{proof}

\begin{lemma} \label{lem:li-p}
    Let $\TT$ and $\TT'$ be topologies with an isomorphism $\beta : \TT \to \TT'$. Let $\val$ be a global state on $\TT'$.
    \begin{enumerate}
        \item For any \QFAshcroft{} formula $\varphi[\nv_1, \ldots, \nv_k]$
    and tuple $\vec{v} \in {\TT'}^k$, if $\val \models \varphi(\vec{v})$, then $\beta^* \val \models \varphi(\beta^{-1}(\vec{v}))$.
        \item For any Ashcroft assertion $\Phi$, if $\val \models \Phi$, then $\beta^* \val \models \Phi$.
    \end{enumerate}
\end{lemma}
\begin{proof}
    Let $\varphi[\nv_1, \ldots, \nv_k]$ be a \QFAshcroft{} formula and tuple $\vec{v} \in {\TT'}^k$. Statement 1 can be shown by induction on the structure of $\varphi$. For instance, as a base case, if $\varphi$ is a formula of the form $\phi[\sval(\nv_1), \ldots, \sval(\nv_k)]$, where $\phi$ is a $\VData$-formula, then
    \begin{align*}
        \val \models \varphi(\vec{v}) &\iff \DD \models \phi(\val(v_1), \ldots, \val(v_k)) \\
        &\iff \DD \models \phi((\beta^*\val)(\beta^{-1}(v_1)), \ldots, (\beta^*\val)(\beta^{-1}(v_k))) \\
        &\iff \beta^* \val \models \varphi(\beta^{-1}(\vec{v})).
    \end{align*}

    Statement 2 follows from 1. For any $\vec{u} \in \TT^k$,
    \begin{align*}
        \val \models \forall \vec{\nv}.\, \varphi(\vec{\nv}) &\implies \val \models \varphi(\beta(\vec{u})) \\
        &\implies \beta^* \val \models \varphi(\beta^{-1}(\beta(\vec{u}))) \\
        &\implies \beta^* \val \models \varphi(\vec{u}).
    \end{align*}
    Hence, $\val \models \forall \vec{\nv}.\, \varphi(\vec{\nv}) \implies \beta^* \val \models \forall \vec{\nv}.\, \varphi(\vec{\nv})$.

\end{proof}

\begin{lemma} \label{lem:li-t}
    Let $\tp{\TT, \sem{-}}$ and $\tp{\TT', \sem{-}'}$ be two programs in a parameterized program such that there is an embedding $f : \TT \into \TT'$. For any global states $\val$, $\val_1$, and $\val_2$ on $\TT'$,
    \begin{itemize}
        \item For any node $v \in \TT$, if $\val_1 \xrightarrow{f(v)} \val_2$, then $\val_1 \circ f \xrightarrow{v} \val_2 \circ f$.
        \item For any parameterized specification, if $\val$ is an initial state, then $\val \circ f$ is an initial state; if $\val|_{f(\TT)}$ contains an error, then $\val \circ f$ contains an error.
    \end{itemize}

\end{lemma}
\begin{proof}
    Let $\val$, $\val_1$, and $\val_2$ be global states on $\TT'$.
    \begin{itemize}
        \item Let $v$ be a node in $\TT$. Restrict $f$ to a local isomorphism $\beta$ from $v$ to $f(v)$. Also restrict $f$ to an injection $g$ from $\setcomp{\nb(v)}$ into $\setcomp{\beta(\nb(v))} = \setcomp{\nb(\beta(v))}$.
    
        Suppose $\val_1 \xrightarrow{f(v)} \val_2$, which means $(\val_1, \val_2)|_{\nb(\beta(v))} \in \sem{\beta(v)}'$ and $\val_1|_{\nb(\beta(v))} = \val_2|_{\nb(\beta(v))}$. Then,
        \begin{gather*}
            (\val_1 \circ f, \val_2 \circ f)|_{\nb(v)} = \beta^*((\val_1, \val_2)|_{\nb(\beta(v))}) \in \beta^* \sem{\beta(v)}' = \sem{v}, \\
            (\val_1 \circ f)|_{\setcomp{\nb(v)}} = \val_1|_{\setcomp{\nb(\beta(v))}} \circ g = \val_2|_{\setcomp{\nb(\beta(v))}} \circ g = (\val_2 \circ f)|_{\setcomp{\nb(v)}}.
        \end{gather*}
        Hence, $\val_1 \circ f \xrightarrow{v} \val_2 \circ f$.

        \item Let $\{\tp{\init_i, \error_i}\}_{i\in\cI}$ be a parameterized specification. Suppose $v \in \TT$ satisfies $\val|_{\nb(f(v))} \in \error_i(f(v))$. Restrict $f$ to a local isomorphism $\beta$ from $v$ to $f(v)$. Then
        \[ (\val \circ f)|_{\nb(v)} = \beta^* (\val|_{\nb(f(v))}) \in \beta^* \error_i(\beta(v)) = \error_j(v). \]
        Hence, $\val \circ f$ contains an error.

        The proof for initial states is similar.

    \end{itemize}
\end{proof}

\begin{theorem}[\Cref{thm:ashcroft-basis}]
    Let $\cP = \{ \prog_i \}_{i\in\cI}$ be a parameterized program over a topology family $\cT = \{ \TT_i \}_{i\in\cI}$. For any $k \in \NN$, if $\cT$ contains a basis $\{ \TT_j \}_{j\in\cJ}$ of rank $k + 1$ as a subset, then for any $\Phi$ of width $k$ and parameterized safety specification, $\Phi$ is an Ashcroft invariant for $\cP$ if and only if it is an Ashcroft invariant for $\{ \prog_j \}_{j\in\cJ}$.
\end{theorem}
\begin{proof}
    The forward direction is trivial. For the backward direction, let $k \in \NN$ and suppose $\cT$ contains a basis $\{ \TT_j \}_{j\in\cJ}$ of rank $k + 1$ as a subset. Let $\Phi := \forall \vec{\nv}.\, \varphi[\vec{\nv}]$ be an Ashcroft invariant for $\{ \prog_j \}_{j\in\cJ}$ w.r.t. parameterized specification $\{\tp{\init_i, \error_i}\}_{i\in\cI}$. Let $i \in \cI$. We prove that $\Phi$ is an Ashcroft invariant for $\prog_i$.

    \begin{description}
        \item[Initialization] Let $\val$ be an initial state on $\TT_i$. Let $\vec{v} \in \TT_i$ be an arbitrary tuple. Observe that any basis of rank $k + 1$ is a basis of any rank lower than $k + 1$. Hence, we can take $j \in \cJ$ such that there is an embedding $f : \TT_j \into \TT_i$ such that the image $\TT'_j := f(\TT_j)$ contains $\vec{v}$. Let $\val' = \val|_{\TT'_j}$ and $\beta : \TT_j \to \TT'_j$ be $f$ with codomain restricted to $\TT'_j$.
        \begin{align*}
            & \beta^* \val' \text{ is an initial state on } \TT_j &&\text{\Cref{lem:li-t}} \\
            {}\implies& \beta^* \val' \models \Phi &&\text{The initialization property of $\Phi$} \\
            {}\implies& \beta^*\val' \models \varphi(\beta^{-1}(\vec{v})) &&\beta^{-1}(\vec{v}) \in \TT_j^k \\
            {}\implies& \val' \models \varphi(\vec{v}) &&\text{\Cref{lem:li-p} applied to $\beta^{-1}$} \\
            {}\implies& \val \models \varphi(\vec{v}). &&\text{$\varphi$ has no quantification over nodes}
        \end{align*}

        \item[Continuation] Let $\val_1 \xrightarrow{u} \val_2$ be a transition in $\prog_i$, where $\val_1$ and $\val_2$ are global states on $\TT_i$. We prove that if $\val_1 \models \Phi$, then $\val_2 \models \Phi$.

        Let $\vec{v} \in \TT_i^k$ be an arbitrary tuple. By the definition of a basis of rank $k + 1$, take $j \in \cJ$ such that there is an embedding $f : \TT_j \into \TT_i$ such that the image $\TT'_j := f(\TT_j)$ contains $u$ and $\vec{v}$. Let $\val'_1 = \val_1|_{\TT'_j}$, $\val'_2 = \val_2|_{\TT'_j}$, and $\beta : \TT_j \to \TT'_j$ be $f$ with codomain restricted to $\TT'_j$.
        \begin{align*}
            \val_1 \models \Phi &\implies \val'_1 \models \Phi &&\text{$\Phi$ quantifies universally over nodes} \\
            &\implies \beta^* \val'_1 \models \Phi &&\text{\Cref{lem:li-p}} \\
            &\implies \beta^* \val'_2 \models \Phi &&\text{\Cref{lem:li-t} and the continuation of $\Phi$} \\
            &\implies \beta^* \val'_2 \models \varphi(\beta^{-1}(\vec{v})) && \beta^{-1}(\vec{v}) \in \TT_j^k \\
            &\implies \val'_2 \models \varphi(\vec{v}) &&\text{\Cref{lem:li-p} applied to $\beta^{-1}$} \\
            &\implies\val_2 \models \varphi(\vec{v}). &&\text{$\varphi$ has no quantification over nodes}
        \end{align*}
        \item[Safety] Let $\val$ be a global state on $\TT_i$. Let $v \in \TT_i$ be an arbitrary node. We prove that if $\val \models \Phi$, then $v$ does not witness an error.

        Take $j \in \cJ$ such that there is an embedding $f : \TT_j \into \TT_i$ such that the image $\TT'_j := f(\TT_j)$ contains $v$. Let $\val' = \val|_{\TT'_j}$ and $\beta : \TT_j \to \TT'_j$ be $f$ with codomain restricted to $\TT'_j$. Then
        \begin{align*}
            \val \models \Phi &\implies \val' \models \Phi \\
            &\implies \beta^* \val' \models \Phi \\
            &\implies \text{No error in } \beta^* \val' \\
            &\implies \text{No error in } \val', &&
        \end{align*}
        where the last line follows from \Cref{lem:li-t}.

    \end{description}
\end{proof}

\begin{theorem}[\Cref{thm:ashcroft-fin}]
    Let $k \in \NN$ and $\cT$ be a topology family such that a finite basis of rank $k+1$ contained in $\cT$ is given. For any Ashcroft assertion $\Phi$ of width $k$, parameterized program over $\cT$, and parameterized safety specification, one can compute a $\VData$-sentence that is valid modulo $\DD$ iff $\Phi$ is an Ashcroft invariant.
\end{theorem}
\begin{proof}
    Write the topology family $\cT$ as $\{ \TT_i \}_{i\in\cI}$ and the parameterized program as $\{ \prog_i \}_{i\in\cI}$. Suppose a basis $\{ \TT_j \}_{j\in\cJ}$ is given. Then, by \Cref{thm:ashcroft-basis,prop:ashcroft-htri}, $\Phi$ is an Ashcroft invariant for $\cP$ if and only if the Hoare triples in $H := \bigcup_{j\in\cJ} H_{\Phi}^{\TT_j}$ are all valid. By \Cref{prop:htri-nf,prop:valid-entail}, for each Hoare triple, we can compute a $\VData$-sentence that is valid modulo $\DD$ if and only if the triple is valid. Since $H$ is finite, the conjunction of these sentences satisfies our need.
\end{proof}

\begin{lemma}[\Cref{lem:qftype-fin}]
    Fix a finite vocabulary $\cV$. Let $\cT$ be a class of $\cV$-structures that is uniformly locally finite.
    \begin{enumerate}
        \item There are finitely many quantifier-free $k$-types $T_1, \ldots, T_m$ in $\cT$, and every $T_r$ is definable by a quantifier-free formula $\alpha_r$, i.e., for any $\cV$-structure $\TT$ and $\vec{v} \in \TT^k$, we have $\TT \models \alpha_r(\vec{v}) \iff \qftype(\TT, \vec{v}) = T_r$. Moreover, the formula defining the quantifier-free $k$-type of a given node tuple in $\cT$ is computable.
        \item Any quantifier-free formula with $k$ free variables is equivalent over $\cT$ to a disjunction of a subset of the $\alpha_r$'s.
    \end{enumerate}
\end{lemma}
\begin{proof}
\begin{enumerate}
    \item Let $\TT \in \cT$ and $\vec{v} \in \TT^k$. For each $u \in \nb(\vec{v})$, let $t_u[\vec{\nv}]$ be a $\cV$-term of minimal height such that $t_u(\vec{v}) = u$. Let $\alpha$ be the conjunction of the following formulas:
    \begin{itemize}
        \item $t_{v_i} = \nv_i$ for each $i \in \{ 1, \ldots, k\}$;
        \item $t_{u_1} \neq t_{u_2}$ for all distinct $u_1, u_2 \in \nb(\vec{v})$;
        \item for all relation symbol $R \in \cV$ and $\vec{u} \in \nb(\vec{v})^n$ where $n = \ar(R)$, if $R(\vec{u})$ is true, then $R (t_{u_1}, \ldots, t_{u_n})$, otherwise $\neg R (t_{u_1}, \ldots, t_{u_n})$; and
        \item $f (t_{u_1}, \ldots, t_{u_r}) = t_w$ for all function symbol $f \in \cV$, $\vec{u} \in \nb(\vec{v})^n$ where $n = \ar(f)$, and $w = f(u_1, \ldots, u_n)$.
    \end{itemize}
    One can verify that $\alpha$ defines $\qftype(\TT, \vec{v})$. The height of terms in $\alpha$ is bounded above by $|\nb(\vec{v})|$. By uniform local finiteness, up to logical equivalence, there are finitely many possible choices of $\alpha$, and therefore finitely many quantifier-free $k$-types.
    \item Combining 1 and \Cref{prop:qftype-sym}, we have that any $\simeq$-equivalence class of node tuples of length $k$ in $\cT$ is definable by some $\alpha_r$.
    
    Let $\phi$ be a quantifier-free formula with $k$ variables. The class of node tuples in $\cT$ defined by $\phi$ is closed under $\simeq$-equivalence. Thus, it is a union of $\simeq$-equivalence classes. Take the disjunction of the $\alpha_r$'s that define each equivalence class.
\end{enumerate}
\end{proof}

\begin{theorem}[\Cref{thm:nf2}]
    Assume the setting of \Cref{thm:nf}. Let $\vec{\chi} \in \{\top,\bot\}^m$ be any indicator vector that selects a basis of rank $k$ from $\cT$. Then over $\cT$, any Ashcroft assertion is equivalent to a normalized Ashcroft assertion of the same width
    \[ \forall \nv_1, \ldots, \nv_k.\, \bigvee_{r=1}^m \alpha_r[\vec{\nv}] \land \phi_r[\sval(t_r^{(1)}[\vec{\nv}]), \ldots, \sval(t_r^{(n_r)}[\vec{\nv}])] \]
    such that for every $r$,  $\phi_r$ is symmetric w.r.t. the type $T_r$, and $\phi_r \equiv \top$ if $\chi_r \equiv \bot$.
\end{theorem}

\begin{proof}
    Recall that for each $r$, $\vec{v}_r$ is an arbitrary tuple that realizes $T_r$, and $\vec{t}_r$ is a list of representative terms. Let $\vec{u}_r = \vec{t}_r(\vec{v}_r)$. Then $\vec{u}_r$, which is a tuple of length $n_r = |\nb(\vec{v}_r)|$, enumerates the nodes in $\AA_r := \nb(\vec{v}_r)$ in some order.

    Take an Ashcroft assertion $\Phi$ of width $k$ in normal form as presented in \Cref{thm:nf}, and set
    \[ \phi'_r[x_1, \ldots, x_{n_r}] := \bigwedge_{s=1}^m \bigwedge_{f : \AA_s \into \AA_r} \phi_s[x_{\pi_f(1)}, \ldots, x_{\pi_f(n_s)}], \]
    where $\pi_f : \{1,\ldots,n_s\} \into \{1,\ldots,n_r\}$ maps each $i$ to the unique $j$ such that $f(\vec{u}_s^{(i)}) = \vec{u}_r^{(j)}$.

    It is easy to see that $\phi'_r$ is symmetric w.r.t. the type $T_r$. Let $g$ be an isomorphism on $\AA_r$. Then,
    \begin{align*}
        \phi'_r[x_{\pi_g(1)}, \ldots, x_{\pi_g(n_r)}] &= \bigwedge_{s=1}^m \bigwedge_{f : \AA_s \into \AA_r} \phi_s[x_{\pi_g(\pi_f(1))}, \ldots, x_{\pi_g(\pi_f(n_s))}] \\
        &= \bigwedge_{s=1}^m \bigwedge_{f : \AA_s \into \AA_r} \phi_s[x_{\pi_{g \circ f}(1)}, \ldots, x_{\pi_{g \circ f}(n_s)}] \\
        &= \bigwedge_{s=1}^m \bigwedge_{f : \AA_s \into \AA_r} \phi_s[x_{\pi_f(1)}, \ldots, x_{\pi_f(n_s)}] \\
        &= \phi'_r[x_1, \ldots, x_{n_r}],
    \end{align*}
    where the second equality is because for all $i = 1, \ldots, n_s$,
    \[
        (g \circ f)(\vec{u}_s^{(i)}) = g(\vec{u}_r^{(\pi_f(i))}) = \vec{u}_r^{(\pi_g(\pi_f(i)))},
    \]
    and thus $\pi_{g \circ f} = \pi_g \circ \pi_f$. Consequently, $\chi_r \to \phi'_r$ is also symmetric w.r.t. the type $T_r$.

    Next, we will show that
    \[ \Phi' := \forall \nv_1, \ldots, \nv_k.\, \bigvee_{r=1}^m \alpha_r[\vec{\nv}] \land (\chi_r \to \phi'_r[\sval(t_r^{(1)}[\vec{\nv}]), \ldots, \sval(t_r^{(n_r)}[\vec{\nv}])]) \]
    is equivalent to the original formula $\Phi$ over $\cT$.

    Let $\TT \in \cT$ and $\val$ be a global state over $\TT$.
    \begin{itemize}
        \item Suppose $\val \models \Phi'$. Let $\vec{w} \in \TT^k$ and suppose it is of type $T_s$. We will show that $\val \models \phi_s[\sval(t_s^{(1)}(\vec{w})), \ldots, \sval(t_s^{(n_s)}(\vec{w}))]$.

        Take the local isomorphism $g$ from $\vec{v}_s$ to $\vec{w}$. By assumption, there is some $r$ with $\chi_r \equiv \top$ and an embedding $h : \AA_r \into \TT$ whose image contains $\nb(\vec{w}) = g(\AA_s)$. Let $f = h'^{-1} \circ g$, where $h'$ is the isomorphism obtained from $h$ by restricting codomain. Then
        \begin{align*}
            \val \models \Phi' &\implies \val \models \phi'_r[\sval(t_r^{(1)}(h(\vec{v}_r))), \ldots, \sval(t_r^{(n_r)}(h(\vec{v}_r)))] \\
            &\implies \val \models \phi_s[\sval(t_r^{(\pi_f(1))}(h(\vec{v}_r))), \ldots, \sval(t_r^{(\pi_f(n_s))}(h(\vec{v}_r)))] \\
            &\implies \val \models \phi_s[\sval(t_s^{(1)}(\vec{w})), \ldots, \sval(t_s^{(n_s)}(\vec{w}))],
        \end{align*}
        where the last implication is because for all $i = 1, \ldots, n_s$,
        \begin{align*}
            t_r^{(\pi_f(i))}(h(\vec{v}_r)) &= h(t_r^{(\pi_f(i))}(\vec{v}_r)) = h(\vec{u}_r^{(\pi_f(i))}) \\
            &= h(f(\vec{u}_s^{(i)})) = g(\vec{u}_s^{(i)}) \\
            &= t_s^{(i)}(\vec{w}).
        \end{align*}

        \item Suppose $\val \models \Phi$. Let $\vec{w} \in \TT^k$ and suppose it is of type $T_r$. We will show that for any $s \in \{1,\ldots,m\}$ with an embedding $f : \AA_s \into \AA_r$, we have $\val \models \phi_s[\sval(t_r^{(\pi_f(1))}(\vec{w})), \ldots, \sval(t_r^{(\pi_f(n_s))}(\vec{w}))]$.
        
        Take the local isomorphism $h$ from $\vec{v}_r$ to $\vec{w}$. For all $i = 1, \ldots, n_s$,
        \begin{align*}
            t_r^{(\pi_f(i))}(\vec{w}) &= t_r^{(\pi_f(i))}(h(\vec{v}_r)) = h(\vec{u}_r^{(\pi_f(i))}) \\
            &= h(f(\vec{u}_s^{(i)})) = t_s^{(i)}(h(f(\vec{v}_s))).
        \end{align*}
        Hence,
        \begin{align*}
            \val \models \Phi &\implies \val \models \phi_s[\sval(t_s^{(1)}(h(f(\vec{v}_s)))), \ldots, \sval(t_s^{(n_s)}(h(f(\vec{v}_s))))] \\
            &\implies \val \models \phi_s[\sval(t_r^{(\pi_f(1))}(\vec{w})), \ldots, \sval(t_r^{(\pi_f(n_s))}(\vec{w}))].
        \end{align*}
    \end{itemize}
\end{proof}

\newpage
\section{Experimental Data}\label{app:experiments}

\subsection{Benchmarks}\label{app:benchmarks} 
The benchmarks are listed in \Cref{tab:benchmarks} in this page and the next. They are over various topology families: lines (L), rings (R), grids (G), depth-2 variable-branching trees (BDT), binary trees (BT), lines with a binary direction predicate (DL), and rings with a ternary orientation predicate (OR).

\begin{center}
\begin{longtable}{c>{\raggedright\arraybackslash}p{0.22\textwidth}>{\raggedright\arraybackslash}p{0.5\textwidth}}
    \caption{List of benchmarks \label{tab:benchmarks}} \\
\hline
\textbf{No.} & \textbf{Name} & \textbf{Description} \\
\hline
L1  & \benchmark{simple_pipeline} & In a pipeline where every node computes the previous value $+$ 1, the end result is greater than $1$ + the initial value. \\ \hline
L2  & \benchmark{pipeline_atomic} & \Cref{fig:pip} without reset \\ \hline
L3  & \benchmark{pipeline_atomic_reset} & \Cref{fig:pip} \\ \hline
L4  & \benchmark{pipeline_reset} & \Cref{fig:pip} without atomicity \\ \hline
L5  & \benchmark{right_ge_any} & Same program as L1. Verify the end result is no less than any intermediate value (inclusive). \\ \hline
L6  & \benchmark{token_passing_bool} & Token passing on a line \\ \hline
L7  & \benchmark{get_max_leftmost_le_any} & In a pipeline that computes the max, verify the initial value is minimal. \\ \hline
R1  & \benchmark{simple_swap} & Nondeterministically swap the values of left and right variables. \\ \hline
R2  & \benchmark{dec} & Nondeterministically set the variable on the right to the minimum of the left and the right. \\ \hline
R3  & \benchmark{inc} & Nondeterministically set the variable on the right to the maximum of the left and the right. \\ \hline
R4  & \benchmark{token_passing_bool} & Token passing on a ring \\ \hline
R5  & \benchmark{token_passing_bool_2} & Token passing on a ring with 2 tokens \\ \hline
R6  & \benchmark{token_passing_int} & Passing 2 tokens on a ring; use a integer-typed counter \\ \hline
R7  & \benchmark{token_passing_int_2} & Passing 2 tokens on a ring; use a integer-typed counter \\ \hline
G1  & \benchmark{simple_grid_pipeline} & Compute min(up, left) + 1. Verify any intermediate value (inclusive) is no less than the value at the top left corner. \\ \hline
G2  & \benchmark{bottom_right_ge_any} & Compute max(up, left) + 1. Verify the value at the bottom right corner is greater than any intermediate value. \\ \hline
G3  & \benchmark{token_passing_bool} & Token passing on a grid \\ \hline
BDT1 & \benchmark{random_max} & Compute the max of the leaf values by updating the value at the root. Verify the value at the root is always no less than any leaf value that has been committed. \\ \hline
BDT2 & \benchmark{convolution_max} & \cite[Figure 2]{pps} \\ \hline
BDT3 & \benchmark{token_passing_bool} & Token passing on the depth-2 trees \\ \hline
BT1 & \benchmark{bool_token_passing} & Token passing on a perfect binary tree \\ \hline
BT2 & \benchmark{count_ge_height} & The size of an arbitrary binary tree is no less than its height. \\ \hline
BT3 & \benchmark{balanced} & If an arbitrary binary tree is balanced, then the left and right subtree of any node have the same size. \\ \hline
BT4 & \benchmark{balanced_pwr} & If an arbitrary binary tree is balanced, then the left and right subtree of any node have the same size; compute 2 to the power of the height instead of the height. \\ \hline
DL1 & \benchmark{get_max} & In a pipeline that computes the maximum number seen so far, any intermediate result does not exceed the final result in the end. \\ \hline
DL2 & \benchmark{get_max_two_way} & For a line of processes with distinct IDs, the leader election scheme that computes the maximum ID from left to right and passes it back produces a unique leader. \\ \hline
DL3 & \benchmark{trusted_chain} & If the first process receives a message, and every process that receives a message is trusted and sends a message to the next process, if it exists, then every process is trusted when there is no pending message. \cite{FeldmanPIS+2017} \\ \hline
OR1 & \benchmark{leader_election_conflated_channels} & The ring leader election algorithm \cite{ChangR1979} selects the process with the maximum ID as the leader. \cite{PadonMPS+2016} Assume conflated channels. \\ \hline
OR2 & \benchmark{leader_election_capacity1_channels} & The ring leader election algorithm \cite{ChangR1979} selects the process with the maximum ID as the leader. \cite{PadonMPS+2016} Assume bounded channels with capacity 1. \\ \hline
R8 & \benchmark{ring_agreement} & The ring agreement protocol in \cite[Sec. 6.10]{Haziza2015}. When it terminates, any non-initiator receives the value computed in the first round, which is no less than its initial value. \\ \hline
\end{longtable}
\end{center}

\subsection{Data} \label{app:data}
The data is listed in Table \ref{tab:results}. The experiment was carried out in a MacBook Air M1 with 16 GB RAM. Data for BT4 is obtained from Golem. All other rows are obtained from Z3/Spacer.

We give the CHC generation and solving phase each a timeout of 600 sec. The winning configuration, under the {\tool} columns, is determined by Solve Time. Gen Time and Size corresponding to the winning configuration is reported, and ties are broken with Gen Time. When there is no configuration to solve the benchmark, the configuration with shortest Gen Time is used.

\begin{table}
\begin{center}
\begin{tabular}{|l|c|c||c|c|c||c|c|c|}\hline \multicolumn{3}{|c||}{} & \multicolumn{3}{|c||}{\Cref{alg:cl}} & \multicolumn{3}{|c|}{{\sc Mosaic}} \\ \hline \multirow{2}{*}{Benchmark} & \multirow{2}{*}{Topology} & {Inv} & Gen & Solve & \multirow{2}{*}{Size} & Gen & \multirow{2}{*}{Size} & Solve\\ & & width & Time& Time & & Time & & Time \\ \hline
L1    &  \multirow{7}{*}{Line} &1  &0.15 & 0.08 & 6/36   & 0.11 & 3/20 & 0.01 \\  
L2    &                        &2 &4.17 & TO&  52/406   & 0.3 & 7/70 & 2.59 \\ 
L3    &                        &2 & 4.25& 5.24& 52/406 & 1.45 & 16/146 & 0.43 \\ 
L4    &                        &2 & 4.49& TO&  52/1112 & 0.31 & 7/190 & TO \\         
L5    &                        &2 & 4.32& TO&  52/759 & 0.3 & 7/130 & 3.32 \\ 
L6    &                        & 2& 3.97& 10.85& 52/1471 & 1.45 & 16/539 & 2.52 \\ 
L7    &                        &1 & 0.15& 0.02& 6/38  & 0.15 & 5/34 & 0.01 \\ \hline   
R1   &  \multirow{7}{*}{Ring} &1 &  0.11& < 0.01& 2/11 & 0.1 & 1/7 & < 0.01 \\ 
R2    &                        &1&  0.11& 0.01& 2/19 & 0.1 & 1/12 & 0.01 \\ 
R3    &                        &1&  0.11& 0.01& 2/19 & 0.1 & 1/12 & 0.01 \\ 
R4    &                        &2&  1.26& 3.79& 18/477 & 0.22 & 4/122 & 1.85 \\ 
R5    &                        &3&  123.2& -&  163/4518 & 1.17 & 4/152 & TO \\         
R6    &                        &2&  1.18&  1.81& 18/292 & 0.62 & 8/132 & 0.43 \\
R7    &                        &3&  114.08& -&  163/3419& 15.25 & 16/401 & 29.85 \\
R8    &                        &2& 1.4 & 18 & 18/383 & 0.67 & 8/173 & 10.5 \\\hline         
G1    &  \multirow{3}{*}{Grid} &1& 1.73 & 0.2 & 11/330 & 0.66 & 9/252 & 0.15 \\ 
G2    &                        &2& TO & - & - & 97.88 & 63/3273 & TO \\         
G3    &                        &2& TO & - & - & 96.55 & 63/3551 & 60.13 \\  \hline 
BDT1    &  \multirow{3}{*}{BD-Tree} &1& 0.13& 0.02&  4/20 & 0.11 & 2/12 &0.01 \\
BDT2    &                        &1& 0.12& 0.04&  4/20 & 0.1 & 2/12 & 0.02 \\ 
BDT3    &                        &2& 1.13& 1.1& 26/419 & 0.66 & 12/205 & 0.45 \\  \hline 
BT1    &  \multirow{4}{*}{Bin-Tree} &2& 30.24 & 7.71 &  36/1071 & 18.17 & 13/472 & 4.3 \\ 
BT2    &                        &1& 0.21 & 0.01 &  36/1071 & 0.2 & 4/50 & 0.11 \\
BT3    &                        &1& 0.21 & TO &  4/50 & 0.12 & 2/30 & TO \\
BT4    &                        &1& 0.21 & 21.81 &  4/50 & 0.12 & 2/30 & 10.43 \\  \hline
DL1    &  \multirow{3}{*}{D-Line} &1  &10.53 & 1.78 & 4/50 & 3.07 & 17/166 & 0.26 \\           
DL2    &                        &2 &9.66 & 49.8&  4/50 & 3.24 & 16/290 & 12.02 \\          
DL3    &                        &2 & 10.36& 1.54& 56/484 & 0.51 & 7/76 & 0.16 \\ \hline
OR1    &  \multirow{2}{*}{O-Ring} &3& TO& 16.36 & 4/152 & 16.36 & 4/152 & 6.08 \\     
OR2    &                        &3& TO& 16.4 & 4/188 & 16.4 & 4/188 & 32.81 \\
\hline
\end{tabular}
\caption{Experimental Results: Size is \#predicates/\#clauses. Times are in seconds. \label{tab:results}}
\end{center}
\end{table}


}{}

\end{document}